\documentclass[a4paper,10pt]{article}

\usepackage{preamble}

\begin{document}

\maketitle     

\begin{abstract}
    We study the \gls{rdf} for the lossy compression of \gls{dt} \gls{wsacs} Gaussian processes with memory, arising from sampling \gls{ct} \gls{wscs} Gaussian source processes. The importance  of this problem arises as such \gls{ct} processes represent communications signals, and sampling must be applied to facilitate the \gls{dt} processing associated with their compression. Moreover, the physical characteristics of oscillators imply that the sampling interval is incommensurate with the period of the \gls{af} of the physical process, giving rise to  the \gls{dt} \gls{wsacs} model considered. In addition, to reduce the loss, the sampling interval is generally shorter than the correlation length,  and thus, the \gls{dt} process is  correlated as well. 
    The difficulty in the \gls{rdf} characterization follows from the information-instability of \gls{wsacs} processes, which renders the traditional information-theoretic tools inapplicable. In this work we utilize  the information-spectrum framework to characterize the \gls{rdf} when a finite and bounded delay is allowed between processing of subsequent source sequences. This scenario extends our previous works which studied settings without  processing delays or without memory. Numerical evaluations reveal the impact of scenario parameters on the  \gls{rdf} with asynchronous sampling.
\end{abstract}

\glsresetall   

\vspace{-0.2cm}

\section{Introduction}
The repetitive operations applied in the generation schemes for communications signals induce  \gls{ct} \gls{wscs} statistics upon these signals \cite[Sec.~1.1]{gardner1994}, \cite[Sec.~1]{gardner2006}. For facilitating digital processing, the observed \gls{ct} signal is first sampled, resulting in a \gls{dt} signal whose statistics depend on the ratio between the sampling interval and the period of the \gls{ct} \gls{af}: When this ratio is a \emph{rational} number, which is referred to as \emph{synchronous sampling}, the sampled process is a \gls{dt} \emph{\gls{wscs}} process; when the ratio is an \emph{irrational} number, which is referred to as \emph{asynchronous sampling}, the sampled process is a \gls{dt} \emph{\gls{wsacs}} process \cite[Sec.~3]{izzo1996}, \cite[Sec.~3.9]{gardner2006}. Consider, for example, the compress-and-forward relay channel \cite{cover1979, kramer2005}: In this channel, the relay compresses the sampled received signal before forwarding it to the destination \cite{dabora2008, wu2013}. Due to the presence of clock jitter (see \cite{vig1993, azeredo-leme2011}) and as the clocks at the source and at the relay are physically separated, the sampling interval and the period of the \gls{ct}  \gls{af} are typically incommensurate, giving rise to 
asynchronous sampling.

To minimize the loss due to sampling, the sampling interval is typically taken smaller than the maximal autocorrelation length of the \gls{ct} \gls{af}, and thereby adjacent samples are statistically correlated. In such a situation we say that the source has \emph{memory}. Moreover, as many communications signals are (asymptotically) Gaussian (see, e.g., \cite{ochiai2001, wei2010, metzger1987}), it follows that sampled communications signals can be modeled as \gls{dt} \gls{wsacs} Gaussian processes with memory, which highlights the importance of characterizing the \gls{rdf} for this class of processes.

In this work we study the \gls{rdf} for \gls{dt} \gls{wsacs} Gaussian processes with memory. The challenge arises from the nonstationarity and the nonergodicity of these processes, which result in information-instability, see \cite{dobrushin1963}, \cite[Sec.~I]{verdu1994}, which renders conventional information-theoretic arguments, relying on typicality, inapplicable.  Among the two relevant alternative frameworks, \gls{ams} processes \cite{gray1980, fontana1981, faigle2007, gray2011} and the {\em information spectrum} framework \cite{han1997, han2010}, in this work the rate-distortion analysis is carried out based on the latter.


The \gls{rdf} of \gls{dt} \gls{wscs} Gaussian processes was characterized in \cite{kipnis2018}, by transforming a scalar \gls{dt} \gls{wscs} process into an equivalent vector stationary process. This result was used in \cite{abakasanga2020} to characterize the \gls{rdf} for \gls{dt} \emph{memoryless} \gls{wsacs} Gaussian processes, derived within the information-spectrum framework. Recently, using a non-random coding approach,  \cite{nishiara2024} proved the achievability of the \gls{rdf} for a general \gls{dt} process under \emph{fixed-length coding} and  \emph{maximum distortion} proposed in \cite{nomura2015}. The dual model, of capacity of channels with additive \gls{wsacs} Gaussian noise was also considered, where \cite{shlezinger2020} assumed memoryless noise, and  \cite{dabora2023} considered noise with memory. In both works, the analysis was carried out within the information spectrum framework.
In the context of the current problem and model, we derived in \cite{tan2024} the \gls{rdf} for an encoding scenario in which the encoder must compress its incoming sequences without delay between subsequent sequences. This assumption resulted in a characterization expressed as the average of the limits of \glspl{rdf}, where each limit is computed with a non-stationary distribution, which does not lead itself to numerical evaluation. In contrast, in this work we consider the scenario in which a finite and bounded  delay is allowed between the encoding of subsequently sampled sequences. This delay facilitates the statistical independence and the optimality of the initial sampling phases, resulting in a different representation for the \gls{rdf}, through the limit of a sequence of computable \glspl{rdf}.



\textit{\textbf{Main Contributions:}} In this work we characterize the \gls{rdf} for compressing \gls{dt} \gls{wsacs} Gaussian processes with memory, subject to  \gls{mse} distortion. Because of the information-instability of \gls{wsacs} processes, the derivation is carried out within the information-spectrum framework. It is assumed that a finite and bounded delay can be introduced between subsequently sampled sequences. This delay is used to facilitate the statistical independence between subsequent sequences, and synchronize sampling to the optimal initial sampling phase, which minimizes the overall compression rate. This setup builds a bridge between the analog signal domain and the digital processing domain for the compression of communications signals, which is a point-of-view absent from previous works on compression, except for our previous works \cite{abakasanga2020} and \cite{tan2024}. Here we also account for the memory of the sampled process, which requires the introduction of a new proof technique, not present in previous works.

\vspace{-0.1cm}
\textit{\textbf{The rest of this work is organized as follows:}} Sec.~\ref{sec: pre_mod_stat_of_prob} reviews \gls{wscs} processes and rate-distortion theory, formulates the problem and introduces relevant information-spectrum definitions; Sec.~\ref{sec: res} presents the \gls{rdf} result; Sec.~\ref{sec: num_eva_dis} numerically evaluates the \gls{rdf} and discusses the impact of different setup parameters on the \gls{rdf}; and Sec.~\ref{sec: con} concludes the work.

\vspace{-0.15cm}
\section{Preliminaries, Model and Problem Statement}
\label{sec: pre_mod_stat_of_prob}

\vspace{-0.15cm}
\subsection{Notations}
\vspace{-0.2cm}
We denote the sets of real numbers, positive real numbers, rational numbers, integers, non-negative integers and positive integers by $\mR$, $\mRdplus$, $\mQ$, $\mZ$, $\mN$ and $\mNplus$, respectively. \Glspl{rv} (resp., deterministic values) are denoted by uppercase letters, e.g., $X$ (resp., lowercase letters, e.g., $x$). Random processes and functions are denoted by stating the time variable in brackets, using round brackets for \gls{ct} and square brackets for \gls{dt}, e.g., $X(t)$, $t \in \mR$, is a \gls{ct} random process, and $x[i]$, $i \in \mZ$ is a \gls{dt} deterministic function. 
Matrices are denoted by sans serif uppercase letters, e.g., $\mmat{A}$, and $(\mmat{A})_{u, v}$, $u, v \in \mN$, denotes its element in the $u$-th row and the $v$-th column. The transpose of a matrix $\mmat{A}$ is denoted by $\mmat{A}^{T}$. For a square matrix $\mmat{B}$, $\det(\mmat{B})$ and $\tr\{\mmat{B}\}$ denote its determinant and its trace, respectively. $\mmat{B} \succ 0$ denotes it is positive definite. Boldface uppercase (resp., lowercase) letters denote column random (resp., deterministic) vectors, e.g., $\mvec{X}$ (resp., $\mvec{x}$). $\zerovec^{k}$ denotes a column all-zero vector of length $k$. 
$\mvec{X} \sim \mGd(\mmvec{\mvec{X}}, \macvmat{\mvec{X}})$ denotes a real Gaussian column random vector $\mvec{X}$ with a mean vector $\mmvec{\mvec{X}}$ and an autocovariance matrix $\macvmat{\mvec{X}}$. $\mE\{\cdot\}$, $\mVar\{\cdot\}$, $|\cdot|$, $\lfloor \cdot \rfloor$, $\lceil \cdot \rceil$, $\log(\cdot)$, $\Pr(\cdot)$, and $p_{X}(\cdot)$ denote the expectation, the variance, the magnitude, the floor function, the ceiling function, the base-$2$ logarithm, the probability and the \gls{pdf} of a continuous \gls{rv} $X$, respectively. We define $a^{+} \triangleq \max\{0, a\}$ and $j = \sqrt{-1}$. 
The differential entropy and the mutual information are denoted by $h(X)$ and $I(X; Y)$, respectively, where $X$ and $Y$ are real \glspl{rv}.

\vspace{-0.05cm}
\subsection{Wide-Sense Cyclostationary Processes}

We next review several definitions relating to \gls{wscs} processes, beginning with the formal definition of such processes:

\begin{definition}[\gls{wscs} processes {\cite[Def.~17.1]{giannakis1999}, \cite[Sec.~3.2]{gardner2006}}]
    A real \gls{ct} (resp., \gls{dt}) random process $X(t)$, $t \in \mR$ (resp., $X[i]$, $i \in \mZ$) is \emph{\gls{wscs}} if both its mean $m_{X}(t) \triangleq \mE\{X(t)\}$ (resp., $m_{X}[i] \triangleq \mE\{X[i]\}$) and its \gls{af} $c_{X}(t, \lambda) \triangleq \mE\{X(t) \cdot X(t + \lambda)\}$ (resp., $c_{X}[i, \Delta] \triangleq \mE\{X[i] \cdot X[i + \Delta]\}$) are periodic in time $t$ (resp., $i$) with some period $T_{c} \in \mRdplus$ (resp., $N_{c} \in \mN^{+}$) for any lag $\lambda \in \mR$ (resp., $\Delta \in \mZ$), i.e., $c_{X}(t, \lambda) = c_{X}(t + T_{c}, \lambda)$ (resp., $c_{X}[i, \Delta] = c_{X}[i + N_{c}, \Delta]$).
\end{definition}

Next, we define \gls{dt} almost periodic functions as follows:

\begin{definition}[\gls{dt} almost periodic functions {\cite{bohr2018}, \cite[Def.~11]{guan2013}}]
\label{def:almost_periodic_function}
    A real, \gls{dt} deterministic function $f[i]$, $i \in \mZ$, is said to be \emph{almost periodic}, if for any $\epsilon \in \mRdplus$, there exists an associated number $l_{\epsilon} \in \mNplus$ such that for any $\alpha \in \mZ$, there exists $\Delta \in [\alpha, \alpha + l_{\epsilon})$, such that $\sup_{i \in \mZ} |f(i + \Delta) - f(i)| < \epsilon$.
\end{definition}

With Def.~\ref{def:almost_periodic_function} we can define \gls{dt} \gls{wsacs} processes as follows:

\begin{definition}[\gls{dt} \gls{wsacs} processes {\cite[Def.~17.2]{giannakis1999}, \cite[Sec.~3.2.2]{gardner2006}}]
    A real \gls{dt} random process $X[i]$, $i \in \mZ$, is called  \emph{\gls{wsacs}} if both its mean $m_{X}[i]$ and its \gls{af} $c_{X}[i, \Delta]$ are \emph{almost periodic} in time $i$ for any lag $\Delta \in \mZ$.
\end{definition}

\vspace{-0.05cm}

\subsection{Rate-Distortion Theory}
Consider first the definition of a lossy source code, stated as follows:

\begin{definition}[Lossy source code {\cite[Sec.~10.2]{cover2006}, \cite[Sec.~3.6]{el_gamal2011}}]
    A \emph{lossy source code} $(m, l)$, where $m$ is the size of the \emph{message set} and $l$ is the \emph{blocklength}, consists of: An \emph{encoder} $f_{l}(\cdot)$, that maps a block of $l$ source symbols $\{x[i]\}_{i = 0}^{l - 1} \equiv \xvecl$, over corresponding alphabets $\{\mset{X}_{i}\}_{i = 0}^{l - 1} \equiv \Xsetl$, into an index selected from a message set of size $m$, i.e., $f_{l}(\cdot): \Xsetl \mapsto \{0, 1, \cdots, m - 1\}$; and a \emph{decoder} $g_{l}(\cdot)$, that assigns a block of $l$ reconstruction symbols $\{\mr{x}[i]\}_{i = 0}^{l - 1} \equiv \xrvecl$, over corresponding alphabets $\{\mrset{X}_{i}\}_{i = 0}^{l - 1} \equiv \rXsetl$, to each received index, i.e., $g_{l}(\cdot): \{0, 1, \cdots, m - 1\} \mapsto \rXsetl$, where  $\frac{1}{l} \log m \triangleq R$ is called the {\em code rate}.
\end{definition}

The mismatch between the source symbol $x$ and its reconstruction $\mr{x}$ is measured using a \emph{distortion function} $d(x, \mr{x})$; the distortion between a block of $l$ source symbols 
and its block reconstruction 
is defined as $\maa{d}(\xvecl, \xrvecl) \triangleq \frac{1}{l} \sum_{i = 0}^{l - 1} d\big(x[i], \mr{x}[i]\big)$. 
When considering compression of sources with continuous alphabets, a commonly used distortion metric is the \emph{squared-error} defined as $d_{se}(x, \mr{x}) \triangleq (x - \mr{x})^2$. An \emph{achievable} rate-distortion pair is defined as follows:

\begin{definition}[Achievable rate-distortion pair {\cite[Sec.~10.2]{cover2006}, \cite[Sec.~3.6]{el_gamal2011}}]
\label{def: achi_rate_dist_pair}
    For a given distortion constraint $D$, if there exists a sequence of $(2^{lR}, l)$ lossy source code  for which
        \vspace{-0.05cm}
    \begin{equation*}
        \limsup_{l \to \infty} \mE\Bigg\{\maa{d}\bigg(\Xvecl, g_{l}\Bigl(f_{l}\big(\Xvecl\big)\Big)\bigg)\Bigg\} \leq D,
        \vspace{-0.05cm}
    \end{equation*}
    then the rate-distortion pair $(R, D)$ is said to be \emph{achievable}.
\end{definition}

Finally, the \gls{rdf} is defined as follows:

\begin{definition}[\gls{rdf} {\cite[Sec.~IV-A]{berger1998}, 
\cite[Sec.~3.6]{el_gamal2011}}]
\label{def: rdf}
    Given a distortion constraint $D$, the \gls{rdf} $R(D)$ is the infimum of all code rates $R$ for which the rate-distortion pair $(R, D)$ is achievable.
\end{definition}

\subsection{Problem Formulation}
\label{sec: prob_form_mod}
Consider a \emph{zero-mean} \gls{ct} \gls{wscs} Gaussian source process $X_{c}(t)$, $t \in \mR$, with an \gls{af} $c_{X_{c}}(t, \lambda)$
where $\lambda \in \mR$ denotes the lag. $c_{X_{c}}(t, \lambda)$ is \emph{uniformly continuous} and \emph{bounded} in $t, \lambda \in \mR$, 
and has a period of $T_{c} \in \mRdplus$ in $t$: $c_{X_{c}}(t, \lambda) = c_{X_{c}}(t + T_{c}, \lambda)$, $|c_{X_{c}}(t, \lambda)| \leq \gamma\in\mR$, $\forall t, \lambda \in \mR$. The random process $X_{c}(t)$ is a \emph{finite-memory} process with a maximal autocorrelation length $\lambda_{c} \in \mRdplus$, i.e., $c_{X_{c}}(t, \lambda) = 0$, $\forall |\lambda| > \lambda_{c}$. $X_{c}(t)$ is uniformly sampled with the sampling interval $T_{s}(\epsilon) \triangleq \frac{T_{c}}{p + \epsilon}$, where $p \in \mNplus$ and $\epsilon \in [0, 1)$. The sampled process is $\Xepphi[i] \triangleq X_{c}\bigl(i \cdot T_{s}(\epsilon) + \phi_{s}\bigl)$, where $i \in \mN$ and $\phi_{s} \in [0, T_c)$ denotes the \emph{initial sampling phase}. The sampling interval satisfies $T_{s}(\epsilon) 
< \lambda_{c}$ which implies that $c_{\Xepphi}[i, \Delta] = 0$, $\forall |\Delta| \geq \Big\lceil \frac{(p + 1) \cdot \lambda_{c}}{T_{c}}\Big\rceil \triangleq \tau_{c} < \infty$. Thus,  $\Xepphi[i]$ is a \emph{finite-memory} process with a maximal autocorrelation length $\tau_{c}$.

The statistics of $\Xepphi[i]$ depend on the ratio between $T_{s}(\epsilon)$ and $T_{c}$: When $\epsilon \in \mQ$, i.e., $\exists u, v \in \mNplus$, s.t. $\epsilon = \frac{u}{v}$, then $\Xepphi[i]$ is a \gls{wscs} process with a period of statistics $N_{c} = p \cdot v + u \triangleq p_{u, v}$. This is referred to as \emph{synchronous sampling}; when $\epsilon \notin \mQ$, 
the sampled process is a \gls{wsacs} process. This is referred to as \emph{asynchronous sampling}.
In this work, a {\em finite and bounded delay between the processing} of subsequently sampled sequences is allowed. 
This delay facilitates the synchronization of the initial sampling phase of every sequence to the optimal phase within $[0, T_{c})$, in the sense of minimizing the overall compression rate. This setup differs from our previous work \cite{tan2024}, in which processing delay was not allowed. 

\subsection{Relevant Information-Spectrum Definitions}
In this work we use the limit superior in probability, which is defined next:

\begin{definition}[Limit superior in probability {\cite[Def.~1.3.1]{han2010}}]
\label{def: limsup_in_prob}
    For a sequence of real \glspl{rv} $\{X_{i}\}_{i = 0}^{\infty}$, its \emph{limit superior in probability} is defined as
    \begin{equation}
        \limsupp_{i \to \infty} {X_{i}} \triangleq \inf\Big\{\alpha \in \mR | \lim_{i \to \infty} \Pr\{X_{i} > \alpha\} = 0\Big\} \triangleq \alpha_{0}.
    \end{equation}
\end{definition}
    \vspace{-0.3cm}

The spectral sup-mutual information rate is now defined as follows:

\begin{definition}[Spectral sup-mutual information rate {\cite[Def.~3.5.2]{han2010}}]
    The \emph{spectral sup-mutual information rate} of two sequences of real continuous \glspl{rv}, $\{X[i]\}_{i = 0}^{l - 1} \equiv \Xvecl$ and $\{Y[i]\}_{i = 0}^{l - 1} \equiv \Yvecl$, is defined as
    \begin{equation*}
        \ssmir(\Xvecinf, \Yvecinf) \triangleq \limsupp_{l \to \infty} \frac{1}{l} \log \frac{p_{\Yvecl | \Xvecl} (\Yvecl | \Xvecl)}{p_{\Yvecl}(\Yvecl)}.
    \end{equation*}
\end{definition}
    \vspace{-0.2cm}

In our proof of the main result we use an \gls{rdf} characterization for arbitrary \gls{dt} processes subject to a uniform integrability condition on the loss function. Uniform integrability is defined as follows:


\begin{definition}[Uniform integrability {\cite[Eqn.~(25.10)]{Bill86probabilitymeasure}}]
    A sequence of real  \glspl{rv} $\{X_{i}\}_{i = 0}^{\infty}$, with a common probability measure $P$  is said to be \emph{uniformly integrable} if $\lim_{u \to \infty} \sup_{i \geq 0} \int_{|X_i|\geq u} |X_i| \mtxt{d} P = 0$. 
\end{definition}

Using the definitions above, it follows that with \emph{fixed-length coding} and with the \emph{average distortion criterion}, the \gls{rdf} for compressing an arbitrary \gls{dt} process is stated as follows:

\begin{theorem}[The \gls{rdf} for an arbitrary \gls{dt} process {\cite[Thm.~5.5.1]{han2010}}]
\label{thm: rdf_arb_dt_proc}
    Consider an arbitrary, real-valued \gls{dt} process $X[i]$, $i \in \mN$. Let $\{X[i]\}_{i = 0}^{l - 1} \equiv \Xvecl$ and $\{\mr{X}_{i}\}_{i = 0}^{l - 1}\equiv \rXvecl$ denote the blocks of $l$ source symbols 
    and of $l$ reconstruction symbols, respectively, and let  
    $F_{\Xvecl, \rXvecl}$ denote their joint \gls{cdf}. If 
    there exists a \emph{deterministic reference word} $\{r_{i}\}_{i = 0}^{l - 1} \equiv \rvecl$, such that the sequence $\{\maa{d}(\Xvecl, \rvecl)\}_{l = 1}^{\infty}$ is \emph{uniformly integrable}, then 
    the \gls{rdf} for compressing $X[i]$ is
    \begin{equation}
    \label{eqn: rdf_arb_dt_sou_proc}
        R(D) = \inf_{\substack{F_{\Xvecinf, \rXvecinf}: \\ \limsup_{l \to \infty} \mE\{\maa{d}(\Xvecl, \rXvecl)\} \leq D}} \ssmir\big(\Xvecinf, \rXvecinf\big).
    \end{equation}
        \vspace{-0.45cm}
\end{theorem}

Recalling the Gaussianity of $X_{c}(t)$ and the boundedness of $c_{X}(t, \lambda)$, in Lemma~\ref{lem:Unif_Integ}, 
we establish the uniform integrability of the distortion for the considered scenario, which 
facilitates the use of Thm.~\ref{thm: rdf_arb_dt_proc} in our analysis:

\begin{lemma}
    \label{lem:Unif_Integ}
    Consider a sequence of $l$ real Gaussian \glspl{rv} $\{X_{i}\}_{i = 0}^{l - 1} \equiv \Xvecl$, for which there exists an upper bound for variances of all its elements (i.e., $\exists \alpha < \infty$, s.t. $\mVar\{X_{i}\} \leq \alpha$ for $0 \leq i \leq l - 1$). Then, the sequence of \gls{mse} distortion values between $\Xvecl$ and the all-zero sequence $\zerovec^{l}$ w.r.t. $l$, denoted by $\big\{\maa{d}_{se}(\Xvecl, \zerovec^{l})\big\}_{l = 1}^{\infty}$, is uniformly integrable.
\end{lemma}

\begin{proof}
    The proof is detailed in Appendix~\ref{sec: pf_lem_uni_intg_app}.
\end{proof}

    \vspace{-0.1cm}

\section{Results}
\label{sec: res}

\vspace{-0.1cm}
As detailed in Sec.~\ref{sec: prob_form_mod}, when $\epsilon \in \mQ$, $\Xepphi[i]$ is a \gls{wscs} process, whose \gls{rdf} was derived in \cite[Thm.~1]{kipnis2018}. Let $\epsilon_{n} \triangleq \frac{\lfloor n \cdot \epsilon \rfloor}{n}$, $n \in \mNplus$, and let $T_{s}(\epsilon_{n}) \triangleq \frac{T_{c}}{p + \epsilon_{n}}$ denote the sampling interval. As $\epsilon_{n}\in\mQ$,  sampling is synchronous and the sampled process $\Xepnphi[i] \triangleq X_{c}\big(i \cdot T_{s}(\epsilon_{n}) + \phi_{s}\big) = X_{c}\big(\frac{i \cdot T_{c}}{p + \epsilon_{n}} + \phi_{s}\big)$ is a \gls{wscs} process
with a maximal correlation length upper bounded by 
$\tau_{c} \triangleq \big\lceil \frac{(p + 1) \cdot \lambda_{c}}{T_{c}} \big\rceil \geq \big\lceil \frac{(p + \epsilon_{n}) \cdot \lambda_{c}}{T_{c}} \big\rceil$. 

Consider a $p_{n}$-dimensional \gls{dt} stationary process $\mvec{X}_{\epsilon_{n}, \phi_{s}}^{p_{n}}[i]$, $i \in\mN$, obtained from $\Xepnphi[i]$ by setting its $m$-th subprocess to $\big(\Xvecepnphipn[i]\big)_{m} = \Xepnphi[i \cdot N_{c} + m]$, $m = 0, 1, \cdots, p_{n} - 1$. The autocorrelation matrix of $\mvec{X}_{\epsilon_{n}, \phi_{s}}^{p_{n}}[i]$ is
    \vspace{-0.2cm}
\begin{equation}
    \label{eqn:WSCS_RDF1}
    \CXvecepnphipn[\Delta] \triangleq \mE\bigg\{\Xvecepnphipn[i] \cdot \Big(\Xvecepnphipn[i + \Delta]\Big)^{T}\bigg\},
\end{equation}
and its \gls{psd} matrix is
\begin{equation}
    \label{eqn:WSCS_RDF2}
    \SXvecepnphipn(f) \triangleq \sum_{\Delta \in \mZ} \CXvecepnphipn[\Delta] \cdot e^{-j 2 \pi f \Delta},
        \vspace{-0.1cm}
\end{equation}
for $-\frac{1}{2} \leq f \leq \frac{1}{2}$. We denote the eigenvalues of $\SXvecepnphipn(f)$ in descending order by $\lamepnphimpn(f)$, $0 \leq m \leq p_{n} - 1$. By \cite[Thm.~1]{kipnis2018}, 
the \gls{rdf} for $\Xepnphi[i]$ for a distortion constraint $D$ is
\begin{equation}
\label{eqn:WSCS_RDF}
    R_{\epsilon_{n}}^{\phi_{s}}(D) = \frac{1}{2 p_{n}} \sum_{m = 0}^{p_{n} - 1} \int_{f = -\frac{1}{2}}^{\frac{1}{2}} \!\!\Bigg(\!\log\bigg(\frac{\lamepnphimpn(f)}{\theta}\bigg)\Bigg)^{+} \!\mtxt{d} f,
        \vspace{-0.1cm}
\end{equation}
where $\theta$ is selected such that
\begin{equation}
    D = \frac{1}{p_{n}} \sum_{m = 0}^{p_{n} - 1} \int_{f = -\frac{1}{2}}^{\frac{1}{2}} \min\Big\{\lamepnphimpn(f), \theta\Big\} \mtxt{d} f.
\end{equation}

As the \gls{rdf} $R_{\epsilon_{n}}^{\phi_{s}}(D)$ depends on $\phi_{s} \in [0, T_{c}\}$, we define
\begin{equation}
\label{eqn:WSCS_RDF_opt}
    R_{\epsilon_{n}}(D) \triangleq \min_{\phi_{s} \in [0, T_{c})} R_{\epsilon_{n}}^{\phi_{s}}(D). 
\end{equation}

In the scenario considered in this work it is assumed that delay of up to $\tau_{c} \cdot T_{s}(\epsilon) + T_{c}$ in \gls{ct} is allowed between subsequently sampled sequences. Then, the \gls{rdf} for $\Xepnphi[i]$ can be obtained as follows:

\begin{theorem}
\label{thm: main_thm}
    Consider the scenario in Sec.~\ref{sec: prob_form_mod}. When a delay of up to $\tau_{c} \cdot T_{s}(\epsilon) + T_{c}$ in \gls{ct} between consecutive sampled sequences is allowed. If the \gls{af} of $X_{c}(t)$ satisfies
    \begin{equation}
    \label{eqn: sdd}
        \min_{0 \leq t < T_{c}} \Bigg\{\! c_{X_{c}}(t, 0) - 2 \cdot \tau_{c} \cdot\! \max_{|\lambda| > \frac{T_{c}}{p + 1}}\! \Big\{\big|c_{X_{c}}(t, \lambda)\big|\Big\}\!\Bigg\}\! \geq\! \gamma_{c} \! > 0,
    \end{equation}
    given a distortion constraint $D \leq \gamma_{c}$, the \gls{rdf} for $\Xepnphi[i]$ is
    \begin{equation}
        \label{eqn:WSACS_RDF}
        R_{\epsilon}(D) = \limsup_{n \to \infty} R_{\epsilon_{n}}(D).
    \end{equation}
\end{theorem}

\begin{proof}
    The proof is detailed in Appendix~\ref{sec: pf_main_thm}.
\end{proof}

\section{Numerical Evaluations and Discussion}
\label{sec: num_eva_dis}


Let $\pulfunc(t)$ denote a periodic function  with a period of $1$. Define   
a single period of $\pulfunc(t)$ as follows:
\begin{equation*}
\pulfunc(t) \triangleq
    \begin{cases}
        \frac{t}{t_{rf}} & , t \in [0, t_{rf}) \\
        1 & , t \in [t_{rf}, t_{rf} + t_{dc}) \\
        1 - \frac{t - t_{dc} - t_{rf}}{t_{rf}} & , t \in [t_{rf} + t_{dc}, 2 \cdot t_{rf} + t_{dc}) \\
        0 & , t \in [2 \cdot t_{rf} + t_{dc}, 1)
    \end{cases},
\end{equation*}
where the rise/fall time $t_{rf} = 0.01$ and the \gls{dc} time $t_{dc} \in [0, 0.98]$. Set the period of $c_{X_{c}}(t, \lambda)$ to be $T_{c} = 5 \microsec$, and define the normalized initial sampling phase $\phi_{s} \in [0, T_{c})$ as $\norisp \triangleq \frac{\phi_{s}}{T_{c}} \in [0, 1)$. Next, let the variance of $c_{X_{c}}(t)$ be defined as $c_{X_{c}}(t, 0) \triangleq 2 + 8 \cdot \pulfunc\big(\frac{t}{T_{c}} - \norisp\big)$. Setting the maximal autocorrelation length of $X_{c}(t)$ to $\lambda_{c} = 4 \microsec$, 
we define $c_{X_{c}}(t, \lambda)$  for $\lambda>0$ as
\begin{equation}
    c_{X_{c}}(t, \lambda) \triangleq
    \begin{cases}
        e^{-\lambda \cdot 10^{6.1}} \cdot c_{X_{c}}(t, 0) & , 0 \leq \lambda \leq \lambda_{c} \\
        0 & , \lambda > \lambda_{c}
    \end{cases}.
\end{equation}
For $\lambda < 0$, $c_{X_{c}}(t, \lambda) = c_{X_{c}}(t + \lambda, -\lambda)$. 
We carry out the numerical evaluations with  $\epsilon = \frac{\pi}{7}$ and $p = 2$. As $\epsilon_{n} \triangleq \frac{\lfloor n \cdot \epsilon\rfloor}{n} \in \mQ$, $\Xepnphi[i] $ corresponds to a \gls{wscs} process with period of statistics $p_{n} \triangleq p \cdot n + \lfloor n \cdot \epsilon \rfloor$, for which  $R_{\epsilon_{n}}^{\phi_{s}}(D)$ is evaluated via Eqn.~\eqref{eqn:WSCS_RDF}. 
Note that for obtaining $R_{\epsilon}(D)$, namely, the \gls{rdf} for the \gls{wsacs} process $\Xepphi[i]$ via Eqn. \eqref{eqn:WSACS_RDF}, 
we verified that Eqn.~\eqref{eqn: sdd} and the condition $D \leq \gamma_{c}$ are satisfied as well.



Figs.~\ref{fig: rdf_vs_n_1} and \ref{fig: rdf_vs_n_2} depict the values of $R_{\epsilon_{n}}^{\phi_{s}}(D)$ as $n$ increases from $1$ to $150$ for $\norisp = 0$ and $\norisp = \frac{\pi}{5}$, respectively, for $D = 0.15$ and $t_{dc} \in \{0.4, 0.7\}$. In both figures, $R_{\epsilon_{n}}^{\phi_{s}}(D)$ is higher when $t_{dc}$ is higher, since a larger time-averaged variance is obtained following a higher $t_{dc}$, which requires more bits per sample to maintain the same distortion (i.e., a higher \gls{rdf}). Observe that when $n$ is small ($n < 25$), $R_{\epsilon_{n}}^{\phi_{s}}(D)$ exhibits significant variations whose pattern significantly depends on $\norisp$. This follows as for small values of $n$, we obtain a larger $T_{s}(\epsilon_{n})$ and a smaller $p_{n}$ consisting of samples sparsely distributed over the period of the variance function of $X_{c}(t)$. This increases the sensitivity of $R_{\epsilon_{n}}^{\phi_{s}}(D)$ to variations in $n$ and to the normalized initial sampling phase $\norisp$. However when $n$ is large ($n \geq 25$), the variation of $R_{\epsilon_{n}}^{\phi_{s}}(D)$ becomes stable and regular. This is because $\lim_{n \to \infty} T_{s}(\epsilon_{n}) = T_{s}(\epsilon)$, which implies that the sampling interval varies very little w.r.t. $n$. Thus, the variation pattern does not exhibit significant variations w.r.t $n$ and $\norisp$. Observe also that due to the nonstationarity of $X_{c}(t)$, $R_{\epsilon_{n}}^{\phi_{s}}(D)$ does not converge to a fixed limiting value as $n$ increases.


\begin{figure}[H]
    \centering
    \includegraphics[scale=0.64]{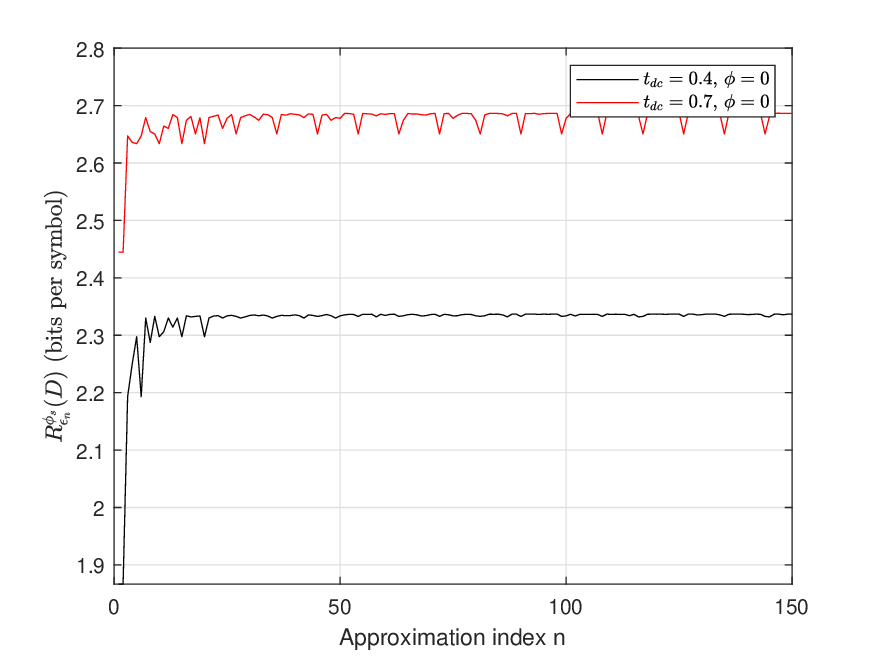}
    \caption{$R_{\epsilon_{n}}^{\phi_{s}}(D)$ versus $n$ for $\phi_{s}=0$.}
    \label{fig: rdf_vs_n_1}
\end{figure}

\begin{figure}[H]
    \centering
    \includegraphics[scale=0.64]{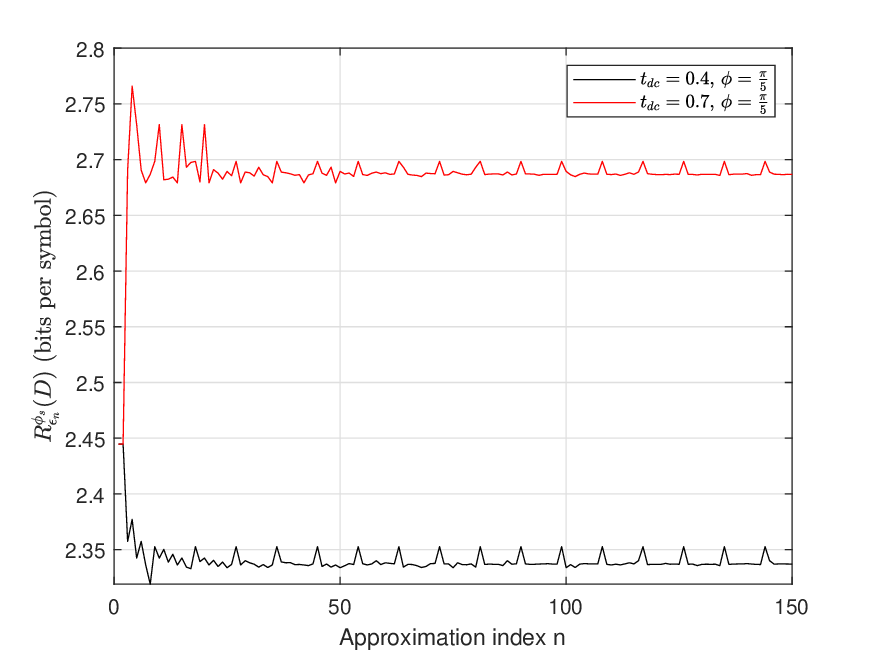}
    \caption{$R_{\epsilon_{n}}^{\phi_{s}}(D)$ versus $n$ for $\phi_{s}=\frac{\pi}{5} T_{c}$.}
    \label{fig: rdf_vs_n_2}
\end{figure}

Fig.~\ref{fig: rdf_vs_nor_ini_samp_pha_1} depicts the variation of $R_{\epsilon_{n}}^{\phi_{s}}(D)$ as $\norisp$ changes from $0$ to $2$ for $t_{dc} = 0.4$, $n \in \{1, 100\}$ and $D = 0.15$. When $n = 1$ ($p_{n} = 2$), $R_{\epsilon_{n}}^{\phi_{s}}(D)$ varies significantly w.r.t. $\norisp$ with a period $1$, which stands in contrast to the case for $n = 100$ ($p_{n} = 244$), where $R_{\epsilon_{n}}^{\phi_{s}}(D)$ varies very little. This observation agrees with the insight from Figs.~\ref{fig: rdf_vs_n_1} and \ref{fig: rdf_vs_n_2}: The asynchronous sampling setup is asymptotically obtained as $n$ is large enough, making $R_{\epsilon_{n}}^{\phi_{s}}(D)$ independent of $\norisp$. Lastly, Fig.~\ref{fig: rdf_vs_D} depicts the variation of $R_{\epsilon_{n}}^{\phi_{s}}(D)$ as $D$ increases from $0.02$ to $0.3$ for $\norisp = \frac{\pi}{5}$, $n = 100$ and $t_{dc} \in \{0.4, 0.7\}$. Observe that $R_{\epsilon_{n}}^{\phi_{s}}(D)$ is a monotonically decreasing convex function w.r.t $D$. This is because for a higher distortion level, less bits per sample are required in the compression.

\begin{figure}[H]
    \centering
    \includegraphics[scale=0.64]{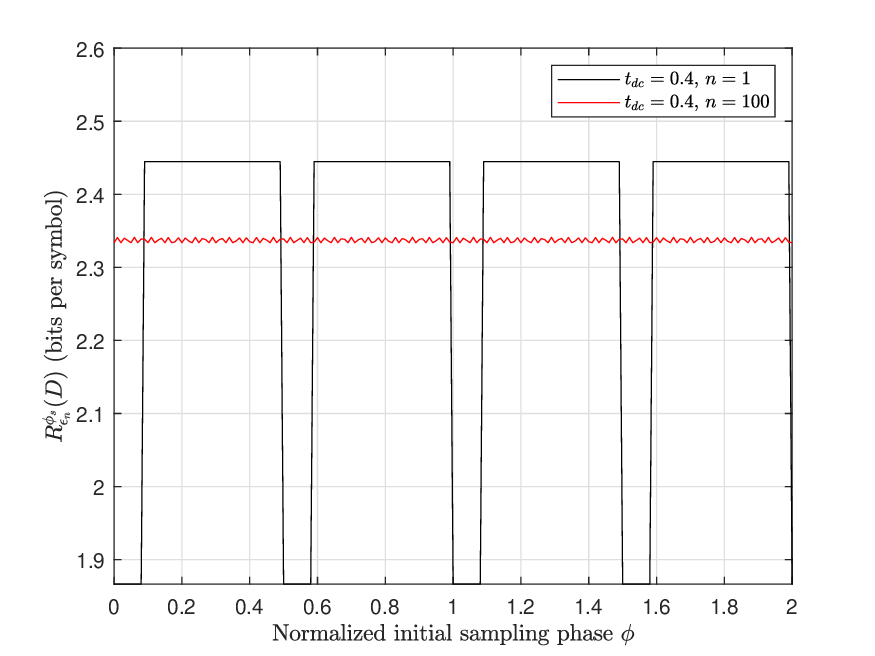}
    \caption{$R_{\epsilon_{n}}^{\phi_{s}}(D)$ versus $\norisp \triangleq \frac{\phi_{s}}{T_{c}}$ for $t_{dc}=0.4$.}
    \label{fig: rdf_vs_nor_ini_samp_pha_1}
\end{figure}

\begin{figure}[H]
    \centering
    \includegraphics[scale=0.64]{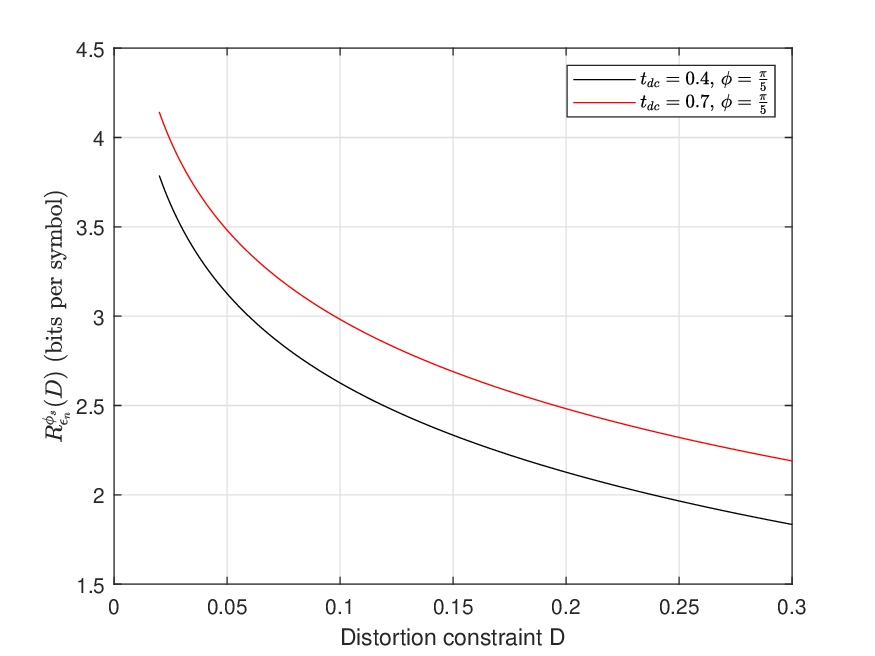}
    \caption{$R_{\epsilon_{n}}^{\phi_{s}}(D)$ versus $D$ for $\phi_{s} = \frac{\pi}{5} T_{c}$.}
    \label{fig: rdf_vs_D}
\end{figure}

\vspace{-0.1cm}
\section{Conclusion}
\label{sec: con}
We have characterized the \gls{rdf} for \gls{dt} \gls{wsacs} Gaussian processes with memory, arising from asynchronously sampling \gls{ct} \gls{wscs} Gaussian source processes. As information-instability of \gls{wsacs} processes renders the conventional information-theoretic tools inapplicable, we employed the information-spectrum framework to derive the \gls{rdf}. In our scenario, a finite and bounded delay between consecutive sampled sequences is allowed, which facilitates having the optimal initial sampling phases at every processed sequence, thereby minimizing the overall compression rate. The resulting  \gls{rdf} is expressed as the limit of a sequence of \glspl{rdf} for synchronous sampling. 
This work demonstrates the relationship between asynchronous sampling, memory and compression rates, which is relevant for facilitating accurate and efficient source coding of communications signals.

\appendix   

\setcounter{lemma}{0}
\renewcommand{\thelemma}{\thesection.\arabic{lemma}}

\setcounter{equation}{0}
\renewcommand{\theequation}{\thesection.\arabic{equation}}

\setcounter{theorem}{0}
\renewcommand{\thetheorem}{\thesection.\arabic{theorem}}

\begin{appendices}

\section{Proof of Lemma~\ref{lem:Unif_Integ}}
\label{sec: pf_lem_uni_intg_app}

Let $\CXvecL$ denote the correlation matrix of $\XvecL$. As $\CXvecL$ is not necessarily a full-rank matrix, we let $\rank\big(\CXvecL\big) = l - \tilde{l}$, where $\tilde{l}$ is the number of degenerate elements in $\XvecL$. The eigenvalue decomposition of the matrix $\CXvecL$\footnote{As $\CXvecL$ is an autocorrelation matrix, it is necessarily symmetric and positive semidefinite and all its eigenvalues are real and nonnegative.} is given as (see \cite[Thm.~11.27]{banerjee2014} and \cite[Sec.~IV-A]{baktash2017})
\begin{equation}
\label{eqn: eig_decp_baktash2017}
    \CXvecL = \begin{bmatrix} \PmatRl & \PmatNl \end{bmatrix} 
    \cdot \begin{bmatrix} 
            \LammattlXvecL & \zeromat^{(l - \tilde{l}) \times \tilde{l}} \\ 
            \zeromat^{\tilde{l} \times (l - \tilde{l})} & \zeromat^{\tilde{l} \times \tilde{l}} 
           \end{bmatrix} \cdot \begin{bmatrix} \PmatRl & \PmatNl \end{bmatrix}^{T},
\end{equation}
where the $(l - \tilde{l}) \times (l - \tilde{l})$ square matrix $\LammattlXvecL$ is a diagonal matrix holding $l - \tilde{l}$ positive eigenvalues, the columns of the $l \times (l - \tilde{l})$ matrix $\PmatRl$ form an orthonormal basis of $\range\big(\CXvecL\big)$, the columns of the $l \times \tilde{l}$ matrix $\PmatNl$ form an orthonormal basis of $\nulspc\big(\CXvecL\big)$, and the $l \times l$ square matrix $\begin{bmatrix} \PmatRl & \PmatNl \end{bmatrix}$ is orthogonal.

Then, we obtain
\begin{equation}
\label{eqn: gauss_dis_zero}
    \begin{bmatrix} \PmatRl & \PmatNl \end{bmatrix}^{T} \cdot \XvecL \eqdist \begin{bmatrix} \mvec{B}^{l - \tilde{l}} \\ \zerovec^{\tilde{l}} \end{bmatrix},
\end{equation}
where
\begin{equation}
\label{eqn: gauss_dis}
    \mvec{B}^{l - \tilde{l}} \sim \mGd\Big(\zerovec^{l - \tilde{l}}, \LammattlXvecL\Big).
\end{equation}
Since $\LammattlXvecL$ is real, symmetric and positive definite, there exists a unique $(l - \tilde{l}) \times (l - \tilde{l})$ square matrix $\Rmatltl$, which is also real, symmetric and positive definite, s.t. $\LammattlXvecL = \Big(\Rmatltl\Big)^{2}$ \footnote{\label{fn: mat_squ_mat}For a Hermitian positive definite matrix $\mmat{A}$, there exists a unique Hermitian positive definite matrix $\mmat{B}$, s.t. $\mmat{A} = \mmat{B}^{2}$ (see \cite[Sec.~1.1]{bhatia2007}, \cite[Thm.~7.2.6-(a)]{horn2012}).}. Following Eqn.~\eqref{eqn: gauss_dis}, we obtain
\begin{equation}
\label{eqn: gamma_gauss}
    \GammavecLtlL \triangleq \Big(\Rmatltl\Big)^{-1} \cdot \mvec{B}^{l - \tilde{l}} \sim \mGd\Big(\zerovec^{l - \tilde{l}}, \idmat^{(l - \tilde{l}) \times (l - \tilde{l})}\Big).
\end{equation}

Next, considering the \gls{mse} distortion between $\XvecL$ and $\zerovec^{l}$, we obtain
    \begin{align}
        & \maa{d}_{se}\Big(\XvecL, \zerovec^{l}\Big) \\
        & = \frac{1}{l} \cdot \big(\XvecL\big)^{T} \cdot \XvecL \\
        \nonumber & \eqdist \frac{1}{l} \cdot \big(\XvecL\big)^{T} \cdot 
        \begin{bmatrix} \PmatRl & \PmatNl \end{bmatrix} \cdot \begin{bmatrix} \PmatRl & \PmatNl \end{bmatrix}^{T} 
        \cdot \XvecL \\
        \nonumber & \labelrel \eqdist{step: gauss_zero} \frac{1}{l} \cdot \big(\XvecL\big)^{T} \cdot \PmatRl \cdot \Big(\PmatRl\Big)^{T} \cdot \XvecL \\
        \nonumber & \eqdist \frac{1}{l} \cdot \big(\XvecL\big)^{T} \cdot \PmatRl \cdot \Big(\Rmatltl\Big)^{-1} \cdot \LammattlXvecL \cdot \Big(\Rmatltl\Big)^{-1} \cdot \Big(\PmatRl\Big)^{T} \cdot \XvecL \\
        \label{eqn: d_eig_chi_squ} & \labelrel \eqdist{step: gamma_gauss} \frac{1}{l} \cdot \Big(\GammavecLtlL\Big)^{T} \cdot \LammattlXvecL \cdot \GammavecLtlL,
    \end{align}
where \eqref{step: gauss_zero} follows from Eqn.~\eqref{eqn: gauss_dis_zero} and \eqref{step: gamma_gauss} follows from Eqn.~\eqref{eqn: gamma_gauss}.

Due to the Gaussianity of the vector $\GammavecLtlL$ and the diagonality of the matrix $\LammattlXvecL$, following Eqn.~\eqref{eqn: d_eig_chi_squ}, we can represent  $\maa{d}_{se}\big(\XvecL, \zerovec^{l}\big)$ as
\begin{equation}
\label{eqn: nor_cros_prod_X}
    \maa{d}_{se}\big(\XvecL, \zerovec^{l}\big) \eqdist \frac{1}{l} \sum_{i = 0}^{l - \tilde{l} - 1} \Big(\LammattlXvecL\Big)_{i, i} \cdot (\gamma_{i})^{2},
\end{equation}
where $\gamma_{i}$, $0 \leq i \leq l - \tilde{l} - 1$, denotes the $i$-th element of the vector $\GammavecLtlL$. Observe that by Eqn.~\eqref{eqn: gamma_gauss}, $\gamma_{i}$, $0 \leq i \leq l - \tilde{l} - 1$, are \gls{iid} standard Gaussian \glspl{rv}, thus $(\gamma_{i})^{2}$, $0 \leq i \leq l - \tilde{l} - 1$, are \gls{iid} central chi-square \glspl{rv} with a single degree of freedom \cite[Sec.~3.8.17]{shynk2012}.
Consequently, from Eqn.~\eqref{eqn: nor_cros_prod_X}, we obtain
    \begin{align}
        \mE\Big\{\maa{d}_{se}\big(\XvecL, \zerovec^{l}\big)\Big\} & = \nonumber \mE\bigg\{\frac{1}{l} \sum_{i = 0}^{l - \tilde{l} - 1} \Big(\LammattlXvecL\Big)_{i,i} \cdot (\gamma_{i})^{2}\bigg\} \\
        & \nonumber = \frac{1}{l} \sum_{i = 0}^{l - \tilde{l} - 1} \Big(\LammattlXvecL\Big)_{i,i} \cdot \mE\big\{(\gamma_{i})^{2}\big\} \\
        & \labelrel ={step: e_chi_squ_rv} \frac{1}{l} \sum_{i = 0}^{l - \tilde{l} - 1} \Big(\LammattlXvecL\Big)_{i,i} \label{eqn: e_dist_seq_up_bou} \labelrel \leq{step: boud_var} \rho,
    \end{align}
where \eqref{step: e_chi_squ_rv} follows as $\mE\big\{(\gamma_{i})^{2}\big\} = 1$, $0 \leq i \leq l - \tl - 1$ \cite[Sec.~3.8.17]{shynk2012} and \eqref{step: boud_var} follows as $\mVar\{X_{i}\} \leq \rho<\infty$ by assumption and as for the square matrix $\CXvecL$, its sum of eigenvalues equals to its trace \cite[Thm.~11.5]{banerjee2014}. 
We therefore upper bound $\mE\Big\{\big|\maa{d}_{se}\big(\XvecL, \zerovec^{l}\big)\big|^{2}\Big\}$ as follows:
    \begin{align*}
        & \mE\Big\{\big|\maa{d}_{se}\big(\XvecL, \zerovec^{l}\big)\big|^{2}\Big\} \\
        & = \mVar\Big\{\maa{d}_{se}\big(\XvecL, \zerovec^{l}\big)\Big\} + \bigg(\mE \Big\{\maa{d}_{se}\big(\XvecL, \zerovec^{l}\big)\Big\}\bigg)^{2} \\
        & \labelrel = {step: explicit upper bound on square} \frac{1}{l^{2}} \sum_{i = 0}^{l - \tilde{l} - 1} \Big(\big(\LammattlXvecL\big)_{i,i}\Big)^{2} \cdot \mVar\big\{(\gamma_{i})^{2}\big\} + \bigg(\mE\Big\{\maa{d}_{se}\big(\XvecL, \zerovec^{l}\big)\Big\}\bigg)^{2} \\
        & \labelrel \leq {step: ind_e_chi_squ} \frac{1}{l^{2}} \sum_{i = 0}^{l - \tilde{l} - 1} \Big(\big(\LammattlXvecL\big)_{i,i}\Big)^{2} \cdot \mVar\big\{(\gamma_{i})^{2}\big\} + {\rho}^{2} \\
        & \labelrel = {step: var_chi_squ_rv} \frac{2}{l^{2}} \sum_{i = 0}^{l - \tilde{l} - 1} \Big(\big(\LammattlXvecL\big)_{i,i}\Big)^{2} + {\rho}^{2} \\
        & \labelrel \leq {step: pos} 2 \cdot \bigg(\frac{1}{l} \sum_{i = 0}^{l - \tilde{l} - 1} \big(\LammattlXvecL\big)_{i,i}\bigg)^{2} + {\rho}^{2} \labelrel \leq {step:upperbound unif integ} 3 {\rho}^{2} < \infty,
    \end{align*}
where \eqref{step: explicit upper bound on square} follows from the statistical independence between $(\gamma_{i})^{2}$, $0 \leq i \leq l - \tilde{l} - 1$; \eqref{step: ind_e_chi_squ} follows from Eqn.~\eqref{eqn: e_dist_seq_up_bou}; \eqref{step: var_chi_squ_rv} follows from $\mVar\big\{(\gamma_{i})^{2}\big\} = 2$, $0 \leq i \leq l - \tl - 1$ \cite[Sec.~3.8.17]{shynk2012}; \eqref{step: pos} follows as $\big(\LammattlXvecL\big)_{i,i}>0$, $0 \leq i \leq l - \tilde{l} - 1$, implying that the square of their sum is not smaller than the sum of their squares; and \eqref{step:upperbound unif integ} follows as $\mVar\{X_{i}\} \leq \rho$ and the sum of eigenvalues of the square matrix $\CXvecL$ equals to its trace \cite[Thm.~11.5]{banerjee2014}. Since this upper bound is independent of $l$, \blk it follows that $\sup_{l \in \mN} \mE\Big\{\big|\maa{d}_{se}\big(\XvecL, \zerovec^{l}\big)\big|^{2}\Big\} < \infty$, thus, by \cite[Sec.~13.3-(a)]{williams1991}, we conclude \blk that the sequence $\Big\{\maa{d}_{se}\big(\XvecL, \zerovec^{l}\big)\Big\}_{l = 1}^{\infty}$ is uniformly integrable\footnote{For a class of \glspl{rv} $\mset{C}$, for some $p > 1$, if $\exists \alpha < \infty$, s.t. $\mE\big\{|X|^{p}\big\} < \alpha$, $\forall X \in \mset{C}$, then $\mset{C}$ is uniformly integrable (see \cite[Sec.~13.3-(a)]{williams1991}).}.

\setcounter{equation}{0}

\section{Proof of Thm.~\ref{thm: main_thm}}
\label{sec: pf_main_thm}

Denote a sequence of $l$ symbols collected from the \gls{dt} \gls{wscs} Gaussian process $\XEpnPhi[i]$ by $\big\{\XEpnPhi[i]\big\}_{i = 0}^{l - 1} \equiv \XvecEpnPhiL$. For proving Thm.~\ref{thm: main_thm}, we first prove two auxiliary lemmas.

\begin{lemma}
\label{lem: convge_autocorr_mats_n}
    The autocorrelation matrix of $\XvecEpnPhiL$ uniformly converges to that of $\XvecEpPhil$ as $n \to \infty$ over $\phi_{s} \in [0, T_{c})$ elementwisely, i.e.,
    \begin{equation}
    \label{eqn: convg_autocorr_mat_n_eps}
        \ulim_{n \to \infty} \Big(\CXvecEpnPhiL\Big)_{u, v} = \Big(\CXvecEpPhiL\Big)_{u, v},
    \end{equation}
    over $\phi_{s} \in [0, T_{c})$ for $0 \leq u, v \leq l - 1$.
\end{lemma}

\begin{proof}
    First, recall that in Section~\ref{sec: res}, $\epsilon_{n} \triangleq \frac{\lfloor n \cdot \epsilon \rfloor}{n}$, $n \in \mNplus$. Thus, it follows that $\frac{n \cdot \epsilon - 1}{n} \leq \epsilon_{n} \leq \frac{n \cdot \epsilon}{n}$, or equivalently, $\epsilon - \frac{1}{n} \leq \epsilon_{n} \leq \epsilon$, and therefore 
    \begin{equation}
    \label{eqn: n_conv_epsilon}
        \lim_{n \to \infty} \epsilon_{n} = \epsilon.
    \end{equation}
    Next, recall that $p \in \mNplus$ and $\epsilon \in [0, 1)$, thus $p + \epsilon_{n} > 0$, $n\in\mNplus$. Then, by Eqn.~\eqref{eqn: n_conv_epsilon}, we obtain $\lim_{n \to \infty} p + \epsilon_{n} = p + \epsilon > 0$. This implies\footnote{If a nonzero real sequence $\{a_{n}\}$, $n \in \mN$, satisfies $\lim_{n \to \infty} a_{n} = a$, $a \neq 0$, then $\lim_{n \to \infty} \frac{1}{a_{n}} = \frac{1}{a}$ \cite[Lemma~9.5]{ross2013}.}
    \begin{equation}
    \label{eqn: p_epsilon_convge}
        \lim_{n \to \infty} \frac{1}{p + \epsilon_{n}} = \frac{1}{p + \epsilon} .
    \end{equation}
    
    As introduced in Section~\ref{sec: prob_form_mod}, the function $c_{X_{c}}(t, \lambda)$ is uniformly continuous in both $t \in \mR$ and $\lambda \in \mR$, therefore it is continuous at the point $\Big(t = \frac{i \cdot T_{c}}{p + \epsilon} + \phi_{s}, \lambda = \frac{\Delta \cdot T_{c}}{p + \epsilon}\Big)$. By Eqn.~\eqref{eqn: p_epsilon_convge} and the definition of continuity\footnote{Sequential criterion for continuity (see \cite[Sec.~5.1.3]{bartle2018}): A real function $f(a): \mset{A} \mapsto \mR$ is continuous at the point $c \in \mset{A}$ if and only if for any real sequence $\{c_{n}\}$, $c_{n}\in\mset{A}$, $n \in \mN$, which satisfies $\lim_{n \to \infty} c_{n} = c$, it is obtained that $\lim_{n \to \infty} f(c_{n}) = f(c)$.}, we have
    \begin{equation}
    \label{eqn: autocorr_func_eps_n_convge}
        \lim_{n \to \infty} c_{X_{c}} \Bigg(\frac{i \cdot T_{c}}{p + \epsilon_{n}} + \phi_{s}, \frac{\Delta \cdot T_{c}}{p + \epsilon_{n}}\Bigg) = c_{X_{c}}\Bigg(\frac{i \cdot T_{c}}{p + \epsilon} + \phi_{s}, \frac{\Delta \cdot T_{c}}{p + \epsilon}\Bigg).
    \end{equation}
    
    Then, we obtain
    \begin{align}
        \lim_{n \to \infty} \Big(\CXvecEpnPhiL\Big)_{u, v} & = \lim_{n \to \infty} c_{X_{c}} \Bigg(\frac{u \cdot T_{c}}{p + \epsilon_{n}} + \phi_{s}, \frac{v \cdot T_{c}}{p + \epsilon_{n}}\Bigg) \\
        & \labelrel ={step: c_n_epsilon_convge} c_{X_{c}}\Bigg(\frac{u \cdot T_{c}}{p + \epsilon} + \phi_{s}, \frac{v \cdot T_{c}}{p + \epsilon}\Bigg) \\
        & \label{eqn: c_epsilon_def} = \Big(\CXvecEpPhiL\Big)_{u, v},
    \end{align}
    for $0 \leq u, v \leq l - 1$, where \eqref{step: c_n_epsilon_convge} follows from Eqn.~\eqref{eqn: p_epsilon_convge}. As  Eqn.~\eqref{eqn: c_epsilon_def} holds $\forall \phi_{s} \in [0, T_{c})$, finally we obtain
    \begin{equation*}
        \lim_{n \to \infty} \max_{\phi_{s} \in [0, T_{c})} \Bigg\{\bigg| \Big(\CXvecEpnPhiL\Big)_{u, v} - \Big(\CXvecEpPhiL\Big)_{u, v} \bigg|\Bigg\} = 0,
    \end{equation*}
    for $0 \leq u, v \leq l - 1$, which corresponds to the definition of uniform convergence \cite[Def.~4.4.3]{trench2022}. It is concluded that $\CXvecEpnPhiL$ uniformly converges to $\CXvecEpPhiL$ as $n \to \infty$ over $\phi_{s} \in [0, T_{c})$ elementwisely, which proves the lemma.
\end{proof}

For $\XvecEpnPhiL$, let  $\{\mr{X}_{\epsilon_{n}}^{\phi_s}[i]\}_{i = 0}^{l - 1}\equiv \rXvecEpnPhiL$ denote the corresponding block of reconstruction symbols, and define two optimal pairs of initial sampling phase and conditional \gls{pdf} as
\begin{subequations}
\label{eqn:OptInf}
    \begin{gather}
        \bigg(\phi_{s, \epsilon_{n}, l}^{\opt}, p\Big(\rXvecEpnPhiLop | \XvecEpnPhiLop\Big)\bigg) \triangleq \label{eqn: op_pair_mut_inf_ep_n} \argmin_{\substack{\Big(\phi_{s} \in [0, T_{c}), p\big(\mr{\mvec{X}}_{\eps_n,\phi_s}^{l}| \mvec{X}_{\eps_n,\phi_s}^{l}\big)\Big): \\ \mE\Big\{\maa{d}_{se}\big(\XvecEpnPhiL, \mr{\mvec{X}}_{\eps_n,\phi_s}^{l}\big)\Big\} \leq D}} \frac{1}{l} I\Big(\XvecEpnPhiL; \mr{\mvec{X}}_{\eps_n,\phi_s}^{l}\Big), \\
        \bigg(\!\!\phi_{s, \epsilon, l}^{\opt}, p\Big(\rXvecEpPhiLop \!| \XvecEpPhiLop\Big)\!\!\bigg)\!\! \triangleq \!\!\!\!\!\!\!\!\label{eqn: op_pair_mut_inf_ep} \argmin_{\substack{\Big(\phi_{s} \in [0, T_{c}), p\big(\mr{\mvec{X}}_{\eps,\phi_s}^{l}| \mvec{X}_{\eps,\phi_s}^{l}\big)\Big) : \\ \mE\Big\{\maa{d}_{se}\big(\XvecEpPhil, \mr{\mvec{X}}_{\eps,\phi_s}^{l}\big) \Big\} \leq D}} \!\!\!\!\frac{1}{l} I\Big(\!\!\XvecEpPhil\!; \!\mr{\mvec{X}}_{\eps,\phi_s}^{l}\!\!\Big).
    \end{gather}
\end{subequations}
Next we define the set $\CsetPhiSvecL$ as
\begin{equation}
\label{eqn: fea_set_phi_dis}
    \CsetPhiSvecL \triangleq \bigg\{\phi_{s} \in [0, T_{c}), \macrmat{\mvec{S}^{l}} \in \mR^{l \times l} \Big| \frac{1}{l} \tr\big(\macrmat{\mvec{S}^{l}}\big) \leq D, \macrmat{\mvec{S}^{l}} \succ 0, \macrmat{\mvec{S}^{l}} = \big(\macrmat{\mvec{S}^{l}}\big)^{T}\bigg\},
\end{equation}
and state a second auxiliary lemma. 
\begin{lemma}
\label{lem: op_mut_inf_converg}
    Consider $\bigg(\phi_{s, \epsilon_{n}, l}^{\opt}, p\Big(\rXvecEpnPhiLop | \XvecEpnPhiLop\Big)\bigg)$ and $\bigg(\phi_{s, \epsilon, l}^{\opt}, p\Big(\rXvecEpPhiLop | \XvecEpPhiLop\Big)\bigg)$ defined in Eqn.~\eqref{eqn:OptInf}. Then, as $n \to \infty$ the sequence of infimums of the objective function in Eqn.~\eqref{eqn: op_pair_mut_inf_ep_n} converges to the infimum of the objective function in Eqn.~\eqref{eqn: op_pair_mut_inf_ep} , i.e.,
    \begin{equation}
    \label{eqn: converg_mut_inf_ep_ep_n}
        \lim_{n \to \infty} \frac{1}{l} I\Big(\XvecEpnPhiLop; \rXvecEpnPhiLop\Big) = \frac{1}{l} I\Big(\XvecEpOpPhiL; \rXvecEpOpPhiLOp\Big).
    \end{equation}
\end{lemma}

\begin{proof}
    Define $\SvecEpnPhiLop \triangleq \XvecEpnPhiLop - \rXvecEpnPhiLop$ and $\SvecEpPhiLop \triangleq \XvecEpOpPhiL - \rXvecEpOpPhiLOp$. Due to the statistical independence between $\SvecEpnPhiLop$ and $\rXvecEpnPhiLop$ and between $\SvecEpPhiLop$ and $\rXvecEpOpPhiLOp$, we obtain
    \begin{subequations}
    \label{eqns: stat_indep_diff_recons}
        \begin{gather}
            \macrmat{\SvecEpnPhiLop} = \macrmat{\XvecEpnPhiLop} - \macrmat{\rXvecEpnPhiLop}, \\
            \macrmat{\SvecEpPhiLop} = \macrmat{\XvecEpOpPhiL} - \macrmat{\rXvecEpOpPhiLOp}.
            \end{gather}
    \end{subequations}
    Define two optimal pairs of initial sampling phase and autocorrelation matrix $\bigg(\phi_{s, \epsilon_{n}, l}^{\opt}, \macrmat{\SvecEpnPhiLop}\bigg) \in \CsetPhiSvecL$ and $\bigg(\phi_{s, \epsilon, l}^{\opt}, \macrmat{\SvecEpPhiLop}\bigg) \in \CsetPhiSvecL$ as
    \begin{subequations}
        \begin{gather}
            \bigg(\phi_{s, \epsilon_{n}, l}^{\opt}, \macrmat{\SvecEpnPhiLop}\bigg)\!\! \triangleq \!\!\label{eqn: argmin_log_det_ep_n} \argmin_{\big(\phi_{s}, \macrmat{\mvec{S}^{l}}\big) \in \CsetPhiSvecL}\!\! \frac{1}{2l} \left(\!\log\left(\frac{\det\Big(\CXvecEpnPhiL\Big)}{\det\Big(\macrmat{\mvec{S}^{l}}\Big)}\right)\right)^{+}\!\!, \\
            \Big(\phi_{s, \epsilon, l}^{\opt}, \macrmat{\SvecEpPhiLop}\Big) \triangleq \label{eqn: argmin_log_det_ep} \argmin_{\big(\phi_{s}, \macrmat{\mvec{S}^{l}}\big) \in \CsetPhiSvecL} \frac{1}{2l} \left(\log\left(\frac{\det\Big(\CXvecEpPhiL\Big)}{\det\Big(\macrmat{\mvec{S}^{l}}\Big)}\right)\right)^{+}.
        \end{gather}
    \end{subequations}
    Then, following the arguments leading to \cite[Eqn.~(B.5a)]{tan2024_2}, we conclude that to prove \eqref{eqn: converg_mut_inf_ep_ep_n}, it is sufficient to show that the sequence of minimums of the objective function in Eqn.~\eqref{eqn: argmin_log_det_ep_n} over $(\phi_{s}, \macrmat{\mvec{S}^{l}}) \in \CsetPhiSvecL$ converges to the minimum of the objective function in Eqn.~\eqref{eqn: argmin_log_det_ep} over $(\phi_{s}, \macrmat{\mvec{S}^{l}}) \in \CsetPhiSvecL$ as $n \to \infty$, i.e.,
    \begin{equation}
    \label{eqn: log_det_op_converg}
        \lim_{n \to \infty} \frac{1}{2l} \log\left(\frac{\det\Bigg(\macrmat{\XvecEpnPhiLop}\Bigg)}{\det\Bigg(\macrmat{\SvecEpnPhiLop}\Bigg)}\right)^{+} = \frac{1}{2l} \log\left(\frac{\det\Bigg(\macrmat{\XvecEpOpPhiL}\Bigg)}{\det\Bigg(\macrmat{\SvecEpPhiLop}\Bigg)}\right)^{+}.
    \end{equation}
    
    In the proof of Eqn.~\eqref{eqn: log_det_op_converg}, we apply \cite[Thm.~2.1]{kanniappan1983}\footnote{\label{footnote: convg_thm} Let $\mset{X}$ and $\mset{Y}$ be two locally convex spaces, where $\mset{Y}$ is also an ordered vector space with a normal order cone. Let $f_{n}: \mset{X} \mapsto \mset{Y}$, $n \in \mNplus$, and $f: \mset{X} \mapsto \mset{Y}$ be continuous and convex mappings. Define $\alpha_{n} \triangleq \inf_{x \in \mset{X}} f_{n}(x)$, $n \in \mNplus$, and $\alpha \triangleq \inf_{x \in \mset{X}} f(x)$. If $\ulim_{n \to \infty} f_{n} = f$, then $\lim_{n \to \infty} \alpha_{n} = \alpha$. See \cite[Sec.~2]{kanniappan1983}. In this scenario, $\mset{X}$ corresponds to $\CsetPhiSvecL$, which is the union of an interval $[0, T_c)$ and the set of real symmetric positive definite matrices satisfying a given trace constraint. $\CsetPhiSvecL$ is a convex space. $\mset{Y}$ corresponds to $\mR$, whose positive cone corresponds to $\{0\}\cup\mRdplus$, which is normal \cite[Example~6.3.5]{pavel2013}, \cite[Footnote~7]{dabora2023}.} and its application requires to prove
    \begin{equation}
    \label{eqn: ramp_func_unif_convg}
        \ulim_{n \to \infty} \frac{1}{2l} \left(\log\left(\frac{\det\Big(\CXvecEpnPhiL\Big)}{\det\big(\macrmat{\mvec{S}^{l}}\big)}\right)\right)^{+} = \frac{1}{2l} \left(\log\left(\frac{\det\Big(\CXvecEpPhiL\Big)}{\det\big(\macrmat{\mvec{S}^{l}}\big)}\right)\right)^{+},
    \end{equation}
    over $\big(\phi_{s}, \macrmat{\mvec{S}^{l}}\big) \in \CsetPhiSvecL$.
        
    Thus, we finally arrive at the uniform convergence condition as follows:
    \begin{equation}
    \label{eqn: logdet_unif_converg}
        \ulim_{n \to \infty} \frac{1}{2l} \log\det\Big(\CXvecEpnPhiL\Big) = \frac{1}{2l} \log\det\Big(\CXvecEpPhiL\Big),
    \end{equation}
    over $\phi_{s} \in [0, T_{c})$.
    
    To prove Eqn.~\eqref{eqn: logdet_unif_converg}, we first show the boundedness of the elements of $\CXvecEpnPhiL$ and of $\CXvecEpPhiL$ and the boundedness of their eigenvalues. The elements of $\CXvecEpnPhiL$ can be upper bounded as follows:
        \begin{align*}
            \bigg| \Big(\CXvecEpnPhiL\Big)_{u, v} \bigg| & = \Big| \mE\big\{\XEpnPhi[u] \cdot \XEpnPhi[v]\big\} \Big| \\
            & \labelrel \leq{step: cauchy_schwarz_ineqn} \sqrt{\mE\Big\{\big(\XEpnPhi[u]\big)^{2}\Big\} \cdot \mE\Big\{\big(\XEpnPhi[v]\big)^{2}\Big\}} \\
            & \labelrel \leq{step: af_bounded_magnitude} \sqrt{\gamma \cdot \gamma} = \gamma,
        \end{align*}
    for $0 \leq u, v \leq l - 1$, where \eqref{step: cauchy_schwarz_ineqn} follows from the Cauchy–Schwarz inequality \cite[Thm.~F.1]{shynk2012} and \eqref{step: af_bounded_magnitude} follows from the boundedness of the \gls{af} of the \gls{ct} \gls{wscs} source process $c_{X_{c}}(t, \lambda)$ for $t, \lambda \in \mR$ (see Section~\ref{sec: prob_form_mod}). As the upper bound $\gamma \in \mRdplus$ is independent of $\epsilon_{n}$,  the elements of $\CXvecEpPhiL$ are similarly upper bounded by $\gamma$.

    Since both $\XvecEpnPhiL$ and $\XvecEpnPhiL$ are sampled from \gls{ct} random processes, their autocorrelation matrices $\CXvecEpnPhiL$ and $\CXvecEpPhiL$ are positive definite and their eigenvalues are all positive \cite[Comment A.1]{dabora2023}. Let the eigenvalues of $\CXvecEpnPhiL$ and of $\CXvecEpPhiL$ arranged in descending order be denoted by $\lambda_{i}^{l}\Big\{\CXvecEpnPhiL\Big\}$ and $\lambda_{i}^{l}\Big\{\CXvecEpPhiL\Big\}$, respectively, $0 \leq i \leq l - 1$.   The maximal eigenvalue of $\CXvecEpnPhiL$ can be upper bounded as follows:
        \begin{equation}
            \maxEig\Big\{\CXvecEpnPhiL\Big\} \leq \tr\Big\{\CXvecEpnPhiL\Big\} \leq \sum_{i=0}^l\lambda_i^l\{\CXvecEpnPhiL\Big\} \labelrel ={step: autocorr_func_boundedness} l \cdot \gamma,
        \end{equation}
    where \eqref{step: autocorr_func_boundedness} follows from the boundedness of the function  $c_{X_{c}}(t, \lambda)$ over $t, \lambda \in \mR$ (see Section~\ref{sec: prob_form_mod}). As the upper bound $l \cdot \gamma$ is independent of $\epsilon_{n}$, similarly the maximal eigenvalue of $\CXvecEpPhiL$ is also upper bounded by $l \cdot \gamma$.
    Now, considering the boundedness of the elements of $\CXvecEpnPhiL$ and of $\CXvecEpPhiL$, the boundedness of their eigenvalues and the fact that $\CXvecEpnPhiL$ uniformly converges to $\CXvecEpPhiL$ as $n \to \infty$ over $\phi_{s} \in [0, T_{c})$ elementwise (as proved in Lemma~\ref{lem: convge_autocorr_mats_n}), then,  by \cite[Thm.~2.4.9.2]{horn2012}\footnote{Let an infinite sequence of $l \times l$ square matrices $\mmat{A}_{n}$, $n \in \mNplus$, be given and suppose $\lim_{n \to \infty} \mmat{A}_{n} = \mmat{A}$ in the elementwise sense. Let $\lambda(\mmat{A}_{n}) = [\lambda_{0}(\mmat{A}_{n}) \ldots \lambda_{l-1}(\mmat{A}_{n})]^{T}$, $n\in\mN$, and $\lambda(\mmat{A}) = [\lambda_{0}(\mmat{A}) \ldots \lambda_{l - 1}(\mmat{A})]^{T}$ be given presentations of the eigenvalues of $\mmat{A}_{n}$, $n \in \mNplus$, and $\mmat{A}$, respectively. Denote the set of all permutations of $\{0, 1, \ldots, l - 1\}$ by $\mset{S}_{l}$. Then, for any $\epsilon > 0$, there exists an associated $N_{\epsilon} \in \mNplus$, such that for all $n \geq N_{\epsilon}$, we have $\min_{\pi \in \mset{S}_{l}} \max_{i = 0, \ldots, l - 1} \{|\lambda_{\pi(i)}\{\mmat{A}_{n}\} - \lambda_{i}\{\mmat{A}\}|\} \leq \epsilon$.} it can be obtained that as  $n \to \infty$,  $\lambda_{i}^{l}\Big\{\CXvecEpnPhiL\Big\}$ convergence to $\lambda_{i}^{l}\Big\{\CXvecEpPhiL\Big\}$  uniform over $\phi_{s} \in [0, T_{c})$ for $0 \leq i \leq l - 1$, i.e., for any $\delta \in \mRdplus$, there exists an associated number $n_{\delta} \in \mNplus$, s.t. for any $n \geq n_{\delta}$, we have
    \begin{equation}
    \label{eqn: unif_converg_eigenvals}
        \bigg|\lambda_{i}^{l}\Big\{\CXvecEpnPhiL\Big\} - \lambda_{i}^{l}\Big\{\CXvecEpPhiL\Big\}\bigg| \leq \delta,
    \end{equation}
    over $\phi_{s} \in [0, T_{c})$ for $0 \leq i \leq l - 1$.

    Next, consider the distance between $\frac{1}{2l} \log\det\Big(\CXvecEpnPhiL\Big)$ and $\frac{1}{2l} \log\det\Big(\CXvecEpPhiL\Big)$: As the determinant of a square matrix equals to the product of its eigenvalues \cite[Proposition~5.2]{ford2014}, we have
        \begin{align}
            & \nonumber \bigg| \frac{1}{2l} \log\det\Big(\CXvecEpnPhiL\Big) - \frac{1}{2l} \log\det\Big(\CXvecEpPhiL\Big) \bigg| \\
            & = \label{eqn: last_exp_unif_converg} \frac{1}{2l} \left| \sum_{i = 0}^{l - 1} \Bigg( \log\bigg(\lambda_{i}^{l}\Big\{\CXvecEpnPhiL\Big\}\bigg) - \log\bigg(\lambda_{i}^{l}\Big\{\CXvecEpnPhiL\Big\}\bigg)\Bigg) \right|.
        \end{align}

    As the logarithmic function is twice differentiable over positive arguments, we use the first-order Taylor series with the remainder of Lagrange form (see \cite[Sec.~20.3]{kline1998}) to express $\log\bigg(\lambda_{i}^{l}\Big\{\CXvecEpnPhiL\Big\}\bigg)$ as
    \begin{align}
    \label{eqn: taylor_expansion}
        & \log\bigg(\lambda_{i}^{l}\Big\{\CXvecEpnPhiL\Big\}\bigg) = \nonumber \log\bigg(\lambda_{i}^{l}\Big\{\CXvecEpPhiL\Big\}\bigg) + \frac{1}{\ln 2} \cdot \frac{1}{\lambda_{i}^{l}\Big\{\CXvecEpPhiL\Big\}} \cdot \bigg(\lambda_{i}^{l}\Big\{\CXvecEpnPhiL\Big\} - \lambda_{i}^{l}\Big\{\CXvecEpPhiL\Big\}\bigg) \\
        & \qquad\qquad\qquad\qquad\qquad + R_{1}\bigg(\lambda_{i}^{l}\Big\{\CXvecEpnPhiL\Big\}\bigg),
    \end{align}
    for $0 \leq i \leq l - 1$, where $R_{1}\bigg(\lambda_{i}^{l}\Big\{\CXvecEpnPhiL\Big\}\bigg)$, the remainder of Lagrange form, is given as (see \cite[Sec.~20.3]{kline1998})
    \begin{equation}
    \label{eqn: lagrange_remainder}
        R_{1}\bigg(\lambda_{i}^{l}\Big\{\CXvecEpnPhiL\Big\}\bigg) = -\frac{1}{2 \cdot \ln 2} \cdot \frac{1}{(\xi_{i})^{2}} \cdot \bigg(\lambda_{i}^{l}\Big\{\CXvecEpnPhiL\Big\} - \lambda_{i}^{l}\Big\{\CXvecEpPhiL\Big\}\bigg)^{2},
    \end{equation}
    in which $\xi_{i}$ satisfies
    \begin{equation}
    \label{eqn: xi_value}
        \min\bigg\{\lambda_{i}^{l}\Big\{\CXvecEpnPhiL\Big\}, \lambda_{i}^{l}\Big\{\CXvecEpPhiL\Big\}\bigg\} \leq \xi_{i} \leq \max\bigg\{\lambda_{i}^{l}\Big\{\CXvecEpnPhiL\Big\}, \lambda_{i}^{l}\Big\{\CXvecEpPhiL\Big\}\bigg\}.
    \end{equation}
    As $\CXvecEpnPhiL$ is a \gls{sdd} matrix (recall Eqn.~\eqref{eqn: sdd} and \cite[Comment~4]{dabora2023}), we can lower bound the minimal eigenvalue of $\CXvecEpnPhiL$ as follows:
    \begin{subequations}
        \begin{align}
            & \minEig\Big\{\CXvecEpnPhiL\Big\} \\
            & \nonumber \labelrel={step: inver_mat_extre_eigvals} \bigg(\maxEig\Big\{\big(\CXvecEpnPhiL\big)^{-1}\Big\}\bigg)^{-1} \\
            & \label{eqn: inv_1_norm_inv_autocorr_mat} \labelrel\geq{step: 1_norm} \bigg(\Big\Vert\big(\CXvecEpnPhiL\big)^{-1}\Big\Vert_{1}\bigg)^{-1} \\
            & \nonumber \labelrel={step: inf_norm} \bigg(\Big\Vert\big(\CXvecEpnPhiL\big)^{-1}\Big\Vert_{\infty}\bigg)^{-1} \\
            & \nonumber \labelrel\geq{step: reason_sdd} \min_{0 \leq u \leq l - 1} \Bigg\{\bigg|\Big(\CXvecEpnPhiL\Big)_{u, v}\bigg| - \sum_{v = 0, v \neq u}^{l - 1}\bigg|\Big(\CXvecEpnPhiL\Big)_{u, v}\bigg|\Bigg\} \\
            & \labelrel\geq{step: sdd_derivation} \min_{0 \leq t < T_{c}} \Bigg\{c_{X_{c}}(t, 0) - 2\tau_{c} \cdot \max_{|\lambda| > \frac{T_{c}}{p + 1}} \big\{|c_{X_{c}}(t, \lambda)|\big\}\Bigg\} \label{eqn: low_boud_min_eig} \labelrel\geq{step: sdd_low_bound} \gamma_{c},
        \end{align}
    \end{subequations}
    where \eqref{step: inver_mat_extre_eigvals} follows from \cite[Thm.~EIM]{beezer2012}; \eqref{step: 1_norm} follows from the symmetry of $(\CXvecEpnPhiL)^{-1}$ and \cite[Eqn.~(4)]{dembo1988}; \eqref{step: inf_norm} follows from the symmetry of $(\CXvecEpnPhiL)^{-1}$; \eqref{step: reason_sdd} follows from \cite[Eqn.~(3)]{moraca2008} and as $\CXvecEpnPhiL$ is a \gls{sdd} matrix; lastly, \eqref{step: sdd_derivation} and \eqref{step: sdd_low_bound} follow from Eqn.~\eqref{eqn: sdd}. As the lower bound $\gamma_{c}$ is independent of $\epsilon_{n}$, similarly we obtain $\minEig\Big\{\CXvecEpPhiL\Big\} \geq \gamma_{c}$. Therefore, $\xi_{i}$ is lower bounded by $\gamma_{c}$, for $0 \leq i \leq l - 1$.

    Plugging Eqn.~\eqref{eqn: taylor_expansion} into  Eqn.~\eqref{eqn: last_exp_unif_converg} and applying  Eqn.~\eqref{eqn: unif_converg_eigenvals}, we obtain that $\forall n \geq n_{\delta}$ and $\forall \phi_{s} \in [0, T_{c})$, it holds that
    \begin{align*}
        & \bigg| \frac{1}{2l} \log\det\Big(\CXvecEpnPhiL\Big) - \frac{1}{2l} \log\det\Big(\CXvecEpPhiL\Big) \bigg| \\
        & \labelrel={step: taylor_expansion} \frac{1}{2l} \Bigg| \sum_{i = 0}^{l - 1} \Bigg(\frac{1}{\ln 2} \cdot \frac{1}{\lambda_{i}^{l}\Big\{\CXvecEpPhiL\Big\}} \cdot \bigg(\lambda_{i}^{l}\Big\{\CXvecEpnPhiL\Big\} - \lambda_{i}^{l}\Big\{\CXvecEpPhiL\Big\}\bigg) + R_{1}\bigg(\lambda_{i}^{l}\Big\{\CXvecEpnPhiL\Big\}\bigg)\Bigg)\Bigg| \\
        & \labelrel\leq{step: ln2_eigenval} \frac{1}{l \cdot \gamma_{c}} \sum_{i = 0}^{l - 1} \bigg|\lambda_{i}^{l}\Big\{\CXvecEpnPhiL\Big\} - \lambda_{i}^{l}\Big\{\CXvecEpPhiL\Big\}\bigg|+ \frac{1}{2l} \sum_{i = 0}^{l - 1} \Bigg| \frac{1}{(\xi_{i})^{2}} \cdot \bigg(\lambda_{i}^{l}\Big\{\CXvecEpnPhiL\Big\} - \lambda_{i}^{l}\Big\{\CXvecEpPhiL\Big\}\bigg)^{2}\Bigg| \\
        & \labelrel\leq{step: unif_converg_delta} \frac{\delta}{\gamma_{c}}\Bigg(1 + \frac{\delta}{2\gamma_{c}}\Bigg),
    \end{align*}
    where \eqref{step: taylor_expansion} follows from Eqns.~\eqref{eqn: last_exp_unif_converg} and \eqref{eqn: taylor_expansion}; \eqref{step: ln2_eigenval} follows from $\frac{1}{\ln 2} < 2$ and $\minEig\Big\{\CXvecEpPhiL\Big\} \geq \gamma_{c}$; and \eqref{step: unif_converg_delta} follows from Eqn.~\eqref{eqn: unif_converg_eigenvals} and $\xi_{i} \geq \gamma_{c}$, for $0 \leq i \leq l - 1$. 
    
    We therefore conclude that $\frac{1}{2l} \log\det\Big(\CXvecEpnPhiL\Big)$ uniformly converges to $\frac{1}{2l} \log\det\Big(\CXvecEpPhiL\Big)$ as $n \to \infty$ over $\phi_{s} \in [0, T_{c})$, which corresponds to Eqn.~\eqref{eqn: logdet_unif_converg}. This facilitates the application of \cite[Thm.~2.1]{kanniappan1983} to conclude Eqn.~\eqref{eqn: converg_mut_inf_ep_ep_n}, which completes the proof of Lemma~\ref{lem: op_mut_inf_converg}.
\end{proof}

We can now state the lemma, which establishes Thm.~\ref{thm: main_thm}, as follows:

\begin{lemma}
\label{lem: final_lem_2nd_scheme}
    For the source sequence generation scheme described in Section~\ref{sec: res}, it holds that
    \begin{equation}
    \label{eqn: rdf_epsilon_n_converg}
        R_{\epsilon}(D) = \limsup_{n \to \infty} R_{\epsilon_{n}}(D),
    \end{equation}
    where $R_{\epsilon_{n}}(D)$ is given in Eqn.~\eqref{eqn:WSCS_RDF_opt}.
\end{lemma}

\begin{proof}
    We first prove the converse part of the lemma as follows:
        \begin{align}
            R_{\epsilon}(D) & \labelrel\geq{step: rdf_geq_optim_mut_inf} \limsup_{l \to \infty} \inf_{\substack{\big(\phi_{s} \in [0, T_{c}), p(\mr{\mvec{X}}^{l} | \XvecEpPhil)\big): \\ \mE\Big\{\maa{d}_{se}\big(\XvecEpPhil, \mr{\mvec{X}}^{l}\big) \Big\}\leq D}} \frac{1}{l} I\big(\XvecEpPhil; \mr{\mvec{X}}^{l}\big) \nonumber \\
            & \nonumber \labelrel={step: min_mut_inf_ep} \limsup_{l \to \infty} \frac{1}{l} I\Big(\XvecEpOpPhiL; \rXvecEpOpPhiLOp\Big) \\
            & \nonumber \labelrel={step: ep_ep_n_equ} \limsup_{l \to \infty} \lim_{n \to \infty} \frac{1}{l} I\Big(\XvecEpnPhiLop; \rXvecEpnPhiLop\Big) \\
            & \nonumber \labelrel={step: lim_exists} \limsup_{l \to \infty} \limsup_{n \to \infty} \frac{1}{l} I\Big(\XvecEpnPhiLop; \rXvecEpnPhiLop\Big) \\
            & \nonumber \labelrel\geq{step: op_mut_inf_to_rdf} \limsup_{l \to \infty} \limsup_{n \to \infty} R_{\epsilon_{n}}^{\phi_{s, \epsilon_{n}, l}^{\opt}}(D) \\
            & \nonumber \labelrel\geq{step: opti_rdf} \limsup_{l \to \infty} \limsup_{n \to \infty} R_{\epsilon_{n}}(D) \\
            & = \label{eqn: converse_2nd_scheme} \limsup_{n \to \infty} R_{\epsilon_{n}}(D),
        \end{align}
    where \eqref{step: rdf_geq_optim_mut_inf} follows from the same arguments leading to \cite[Eqn.~(B.7)]{tan2024_2}; \eqref{step: min_mut_inf_ep} follows by plugging the optimal solution of Eqn.~\eqref{eqn: op_pair_mut_inf_ep}; \eqref{step: ep_ep_n_equ} follows from Lemma~\ref{lem: op_mut_inf_converg}; \eqref{step: lim_exists} follows as the limit of the sequence of optimized mutual information terms in Eqn.~\eqref{eqn: converg_mut_inf_ep_ep_n} exists and is finite, thus its limit superior equals to its limit \cite[Thm.~4.1.12]{trench2022}; \eqref{step: op_mut_inf_to_rdf} follows as $\frac{1}{l} I\Big(\XvecEpnPhiLop; \rXvecEpnPhiLop\Big)$ is an achievable code rate for a fixed sufficiently large blocklength $l$, which is not lower than the \gls{rdf} $R_{\epsilon_{n}}^{\phi_{s, \epsilon_{n}, l}^{\opt}}(D)$ by definition \footnote{Segment the continuously generated \gls{dt} \gls{wscs} source symbols into separated blocks of length $l$, then insert an finite and  bounded intervals between consecutive blocks to facilitate the statistical independence of different blocks and synchronize the initial sampling phase of each block to $\phi_{s, \epsilon_{n}, l}^{\opt}$. All blocks are thus \gls{iid} and their reconstructions follow from the optimal conditional distribution $p\bigg(\rXvecEpnPhiLop \Big| \XvecEpnPhiLop\bigg)$ defined in Eqn.~\eqref{eqn: op_pair_mut_inf_ep_n}. As $\frac{1}{l} I\bigg(\XvecEpnPhiLop; \rXvecEpnPhiLop\bigg)$ is an achievable rate for a fixed $l$, if   $\frac{1}{l} I\bigg(\XvecEpnPhiLop; \rXvecEpnPhiLop\bigg) < R_{\epsilon_{n}}^{\phi_{s, \epsilon_{n}, l}^{\opt}}(D)$, then the overall rate is smaller than $R_{\epsilon_{n}}^{\phi_{s, \epsilon_{n}, l}^{\opt}}(D)$, which contradicts the \gls{rdf} definition in Def.~\ref{def: rdf}. Note that as $l$ is sufficiently large, the decrease in  rate and the increase in distortion due to the above construction become asymptotically negligible (see analysis following Eqns.~\eqref{eqn: overall_code_rate} and \eqref{eqn: cw_r_dist_met}, respectively.)}; and \eqref{step: opti_rdf} follows from Eqn.~\eqref{eqn:WSCS_RDF_opt}.
    
    For the achievability part of the lemma, we consider the source sequence generation scheme described in Section~\ref{sec: res}, in which all source sequences are generated with the optimal initial sampling phases. In this scheme, the sequence of source symbols obtained by asynchronously sampling the \gls{ct} \gls{wscs} source process is equally segmented into multiple blocks. Each segmented block of source symbols has a finite blocklength of $l \in \mNplus$, which is referred to as an \emph{$l$-block}. Then, between consecutive $l$-blocks, a guard interval is inserted in order to facilitate  statistical independence among the $l$-blocks and simultaneously synchronize the start time of the subsequent $l$-block to the optimal initial sampling phase within a single period of the \gls{af} of $X_c(t)$.
    The optimal initial sampling phase value for each $l$-block, $\PhiEpLOp$, is obtained from the minimization in Eqn.~\eqref{eqn: op_pair_mut_inf_ep}. In the following, we elaborate on the operations at the encoder and at the decoder, respectively.
    
    \textit{Encoder's Operations:} The encoder maintains a guard time between processing of consecutive $l$-blocks. This guard time is set to sufficiently long to facilitate statistical independence between symbols belonging to different processed $l$-blocks. Given that the maximal correlation length for the \gls{dt} \gls{wsacs} Gaussian process $\XEpPhi[i]$ is $\tau_{c}$ samples, the duration of the guard interval in \gls{ct} should be at least $\tau_{c} \cdot T_{s}(\epsilon)$. An $l$-block appended with $\tau_{c}$ samples 
    is referred to as an \emph{$(l + \tau_{c})$-block}. Then, an interval of duration $\Delta_{g}$ in \gls{ct} 
    is added to facilitate the synchronization of the start time of the subsequent $l$-block to the optimal initial sampling phase. Therefore, an input codeword of $k \cdot l$ source symbols is transmitted over a \gls{dt} interval whose length corresponds to $k \cdot \Big(l + \tau_{c} + \frac{\Delta_{g}}{T_{s}(\epsilon)}\Big)$ samples. Let $\Delta_{g}'$ denote the sampling phase of the last sample of each $(l + \tau_{c})$-block, 
        $\Delta_{g}' \triangleq \big(\PhiEpLOp + (l + \tau_{c}) \cdot T_{s} (\epsilon)\big) \mod T_{c}$.
    Then $\Delta_{g}$ is given by
    \begin{equation*}
        \Delta_{g} =
        \begin{cases}
            \PhiEpLOp - \Delta_{g}', \quad & \Delta_{g}'  \leq \PhiEpLOp, \\
            T_{c} - \Delta_{g}' + \PhiEpLOp, \quad &\Delta_{g}' > \PhiEpLOp.
        \end{cases}
    \end{equation*}
    It is noted that $\Delta_{g}$ is deterministically computable at the encoder, since the \gls{af} of the \gls{ct} \gls{wscs} source process $X_{c}(t)$ is assumed to be known at the encoder, $T_{s}(\epsilon)$ is the sampling interval at the transmitter, and $\PhiEpLOp$ is computable from the minimization in Eqn.~\eqref{eqn: op_pair_mut_inf_ep}. The scheme detailed above transmits an $l$-block at rate $R$ bits per sample over an interval corresponding to $l + \tau_{c} + \frac{\Delta_{g}}{T_{s}(\epsilon)}$ samples. Thus, the overall code rate of this scheme is
    \begin{equation}
    \label{eqn: overall_code_rate}
        R \cdot \frac{l}{l + \tau_{c} + \frac{\Delta_{g}}{T_{s}(\epsilon)}} = R \cdot \Bigg(1 - \frac{\tau_{c} + \frac{\Delta_{g}}{T_{s}(\epsilon)}}{l + \tau_{c} + \frac{\Delta_{g}}{T_{s}(\epsilon)}}\Bigg).
    \end{equation}
    Note that as $l \to \infty$, the decrease of the code rate due to the introduction of the guard interval becomes asymptotically negligible.

    After inserting guard intervals, all $l$-blocks are statistically independent and have the same initial sampling phases, thus all $l$-blocks are \gls{iid} and accordingly a single optimal codebook can used for compression of all blocks. Each $l$-block with the optimal initial sampling phase $\PhiEpLOp$ (which is denoted as $\XvecEpOpPhiL$) is compressed into a message index using the optimal codebook denoted by $\CBsetLOpPhiOp$, which is generated according to the conditional distribution $ p\Big(\rXvecEpPhiLop | \XvecEpPhiLop\Big)$ obtained through the minimization in Eqn.~\eqref{eqn: op_pair_mut_inf_ep}. This message index is sent to the decoder.
    
    \textit{Decoder's Operations:} 
    The proposed source sequence generation scheme with the optimal codebook $\CBsetLOpPhiOp$ represents an input codeword of $k \cdot (l + \tau_{c})$ source symbols, denoted as $\big\{\XEpOpPhi[i]\big\}_{i = 0}^{k \cdot (l + \tau_{c}) - 1} \equiv \XvecEpOpPhiGIk$, by $k \cdot (l + \tau_{c})$ reconstruction samples, denoted as $\big\{\rOpXEpOpPhi[i]\big\}_{i = 0}^{k \cdot (l + \tau_{c}) - 1} \equiv \rXOpvecEpOpPhiGIk$. The vector $\rXOpvecEpOpPhiGIk$ consists of $k$ reconstructed $(l + \tau_{c})$-blocks and each containing an optimal reconstructed $l$-block and $\tau_{c}$ zero samples. Let $\big\{\XEpOpPhi[i]\big\}_{i = 0}^{k \cdot l - 1} \equiv \XvecEpOpPhik$ denote the set of $k$ segmented $l$-blocks at the encoder and $\big\{\rOpXEpOpPhi[i]\big\}_{i = 0}^{k \cdot l - 1} \equiv \rOpXvecEpOpPhik$ denote the set of $k$ optimal reconstructed $l$-blocks (i.e., after discarding $k \cdot \tau_{c}$ zero samples) at the decoder. Define $\big\{\OpSEpOpPhi[i]\big\}_{i = 0}^{k \cdot (l + \tau_{c}) - 1} \equiv \OpSvecEpOpPhiGIk \triangleq \XvecEpOpPhiGIk - \rXOpvecEpOpPhiGIk$ and $\big\{\OpSEpOpPhi\big\}_{i = 0}^{k \cdot l - 1} \equiv \OpSvecEpOpPhik \triangleq \XvecEpOpPhik - \rOpXvecEpOpPhik$. 
    Note that as the optimal codebook used is generated via \eqref{eqn: op_pair_mut_inf_ep},
    then $\OpSvecEpOpPhiGIk$ belongs of the set $\CsetPhiSvecL$ defined in Eqn.~\eqref{eqn: fea_set_phi_dis}.
    With these definitions, the distortion between $\XvecEpOpPhiGIk$ and $\rXOpvecEpOpPhiGIk$ is upper bounded as follows:
        \begin{align}
            & \nonumber \mE\Bigg\{\frac{1}{k \cdot (l + \tau_{c})} \sum_{i = 0}^{k \cdot (l + \tau_{c}) - 1} \Big(X_{\epsilon,\phisopt}[i] - \rOpXEpOpPhi[i]\Big)^{2}\Bigg\} \\
            & \nonumber \leq \mE\Bigg\{\frac{1}{k \cdot l} \sum_{i = 0}^{k \cdot (l+\tau_c) - 1} \Big(\OpSEpOpPhi[i]\Big)^{2}\Bigg\} \\
            & \nonumber \labelrel={step: zero_samples} \frac{1}{l} \sum_{l' = 0}^{l - 1} \Bigg(\frac{1}{k} \sum_{k' = 0}^{k - 1} \mE\bigg\{\Big(S_{\epsilon, \phisopt}^{\opt}[k' \cdot (l + \tau_{c})+ l']\Big)^{2}\bigg\}\Bigg) \\
            & \quad + \frac{1}{l} \sum_{l' = 0}^{\tau_{c} - 1} \Bigg(\frac{1}{k} \sum_{k' = 0}^{k - 1} \mE\bigg\{\Big(X_{\epsilon,\phisopt}[k' \cdot (l + \tau_{c}) + l + l']\Big)^{2}\bigg\}\Bigg) \\
            & \labelrel={step: iid_l_block} \frac{1}{l} \sum_{l' = 0}^{l - 1} \mE\bigg\{\Big(S_{\epsilon, \phisopt}^{\opt}[l']\Big)^{2}\bigg\} + \frac{1}{l} \sum_{l' = 0}^{\tau_{c} - 1} \Bigg(\mE\bigg\{\Big(X_{\epsilon,\phisopt}[l + l']\Big)^{2}\bigg\}\Bigg) \\
            & \label{eqn: cw_r_dist_met} \labelrel\leq{step: dist_met} D + \frac{\tau_{c} \cdot \gamma}{l},
        \end{align}
    where \eqref{step: zero_samples} follows as the $k \cdot \tau_{c}$ source symbols used in guard intervals at the encoder are reconstructed as $k \cdot \tau_{c}$ zero samples at the decoder; \eqref{step: iid_l_block} follows as all $l$-blocks are \gls{iid} and they are reconstructed using the same codebook; and \eqref{step: dist_met} follows from the trace constraint condition in the definition of the set $\CsetPhiSvecL$ in Eqn.~\eqref{eqn: fea_set_phi_dis} and as the \gls{af} of the \gls{ct} \gls{wscs} Gaussian source process $X_{c}(t)$ is bounded by $\gamma$, see Sec.~\ref{sec: prob_form_mod}. This analysis implies that the compression of $\XvecEpOpPhiGIk$ asymptotically satisfies the given distortion constraint $D$ as $l \to \infty$.
    
    In the following, denote the $m$-th $l$-block with the optimal initial sampling phase by $\XvecEpOpPhim$, denote its optimal reconstruction by $\rXvecEpOpPhimOp$, and let $\FXrXOpmOpPhi$ denote their joint \gls{cdf}, $0 \leq m \leq k - 1$. The mutual information density rate between the $m$-th $l$-block and its optimal reconstruction is defined as
    \begin{equation}
    \label{eqn: midr_lbk_opr}
        Z\Bigg(\FXrXOpmOpPhi\Bigg) \triangleq \frac{1}{l} \log \left(\frac{p_{\XvecEpOpPhim \big| \rXvecEpOpPhimOp}\Big(\XvecEpOpPhim \big| \rXvecEpOpPhimOp\Big)}{p_{\XvecEpOpPhim}\Big(\XvecEpOpPhim\Big)}\right).
    \end{equation}
    Next, let the mutual information density rate between $\XvecEpOpPhik$ and $\rOpXvecEpOpPhik$ be denoted as $\ZFXrXOpklOpPhi$:
    \begin{equation}
        \ZFXrXOpklOpPhi \\
        \label{eqn: midr_cw_r_opisp} \triangleq \frac{1}{k \cdot l} \log \left(\frac{p_{\XvecEpOpPhik \Big| \rOpXvecEpOpPhik}\bigg(\XvecEpOpPhik \Big| \rOpXvecEpOpPhik\bigg)}{p_{\XvecEpOpPhik}\bigg(\XvecEpOpPhik\bigg)}\right).
    \end{equation}
    By Eqn.~\eqref{eqn: cw_r_dist_met},  the distortion associated with compressing $\XvecEpOpPhiGIk$ is upper bounded by $D + \frac{\tau_{c} \cdot \gamma}{l}$, which asymptotically approaches $D$ as $l \to \infty$. We can therefore upper bound $R_{\epsilon}(D)$ as follows:
        \begin{equation}
            R_{\epsilon}(D) \labelrel\leq{step: rdf_leq_limsupp} \limsupp_{l \to \infty} \ZFXrXOpklOpPhi \label{eqn: 2nd_gen_sche_achi_fin_exp} \labelrel\leq{step: limsupp_leq_rdf_ep_n} \limsup_{n \to \infty} R_{\epsilon_{n}}(D),
        \end{equation}
    where \eqref{step: rdf_leq_limsupp} follows from the definition of rate-distortion pairs in Def.~\ref{def: achi_rate_dist_pair} and Eqn.~\eqref{eqn: rdf_arb_dt_sou_proc}, which imply that the rate-distortion pair $\Bigg(\limsupp_{l \to \infty} \ZFXrXOpklOpPhi, D\Bigg)$ is achievable and 
    $\limsupp_{l \to \infty} \ZFXrXOpklOpPhi$ 
    cannot be smaller than the \gls{rdf} $R_{\epsilon}(D)$. Next, we show the inequality for step~\eqref{step: limsupp_leq_rdf_ep_n}.
    
    Following Eqn.~\eqref{eqn: midr_cw_r_opisp}, we obtain
        \begin{align*}
            & \ZFXrXOpklOpPhi \\
            & \triangleq \frac{1}{k \cdot l} 
            \left(\frac{p_{\XvecEpOpPhik \Big| \rOpXvecEpOpPhik}\bigg(\XvecEpOpPhik \Big| \rOpXvecEpOpPhik\bigg)}{p_{\XvecEpOpPhik}\bigg(\XvecEpOpPhik\bigg)}\right) \\
            & \labelrel={step: mlbk_mut_indp} \frac{1}{k \cdot l} \log \left(\prod_{m = 0}^{k - 1} \frac{p_{\XEpOpPhim \big| \rXEpOpPhimOp}\Big(\XvecEpOpPhim | \rXvecEpOpPhimOp\Big)}{p_{\XEpOpPhim}\Big(\XvecEpOpPhim\Big)}\right) \\
            & = \frac{1}{k} \sum_{m = 0}^{k - 1} \frac{1}{l} \log \left(\frac{p_{\XEpOpPhim \big| \rXEpOpPhimOp}\Big(\XvecEpOpPhim \big| \rXvecEpOpPhimOp\Big)}{p_{\XEpOpPhim}\Big(\XvecEpOpPhim\Big)}\right) \\
            & \labelrel={step: midr_lbk_opr} \frac{1}{k} \sum_{m = 0}^{k - 1} Z\bigg(\FXrXOpmOpPhi\bigg),
        \end{align*}
    where \eqref{step: mlbk_mut_indp} follows from the statistical independence between different $l$-blocks and \eqref{step: midr_lbk_opr} follows from the definition in Eqn.~\eqref{eqn: midr_lbk_opr}. Taking the expectation and the variance of $\ZFXrXOpklOpPhi$, we have
        \begin{align}
            & \mE\left\{\ZFXrXOpklOpPhi\right\} \\
            & = \frac{1}{k} \sum_{m = 0}^{k - 1} \mE\Bigg\{Z\bigg(\FXrXOpmOpPhi\bigg)\Bigg\} \label{eqn: mean_density_rates} \labelrel={step: equ_nor_mut_inf} \frac{1}{l} I\Big(\XvecEpOpPhiL; \rXvecEpOpPhiLOp\Big), \\
            & \mVar\left\{\ZFXrXOpklOpPhi\right\} \\
            & \labelrel={step: midr_mut_indp} \frac{1}{k^{2}} \sum_{m = 0}^{k - 1} \mVar\Bigg\{Z\bigg(\FXrXOpmOpPhi\bigg)\Bigg\} \label{eqn: up_bou_density_rates} \labelrel<{step: prev_var_midr} \frac{3}{k \cdot l},
        \end{align}
    where \eqref{step: equ_nor_mut_inf} follows from the notion of \cite[Eqns.~(B.5a) and (B.15)]{tan2024_2}; \eqref{step: midr_mut_indp} follows from the statistical independence between different $l$-blocks, which induces the statistical independence between mutual information density rates; and \eqref{step: prev_var_midr} follows from the similar derivation leading to the upper bound in \cite[Eqn.~(B.17)]{tan2024_2}.
    
    Next, plugging the expectation in Eqn.~\eqref{eqn: mean_density_rates} and the upper bound of the variance in Eqn.~\eqref{eqn: up_bou_density_rates} into Chebyshev inequality \cite[Eqn.~(1.58)]{gallager2013}, we obtain
    \begin{equation}
        \Pr\Bigg\{\Bigg|\ZFXrXOpklOpPhi - \frac{1}{l} I\Big(\XvecEpOpPhiL; \rXvecEpOpPhiLOp\Big)\Bigg| \geq \frac{1}{(k \cdot l)^{\frac{1}{3}}}\Bigg\} < \frac{3}{(k \cdot l)^{\frac{1}{3}}},
    \end{equation}
    where the upper bound of the probability decreases as $k \cdot l$ increases. Therefore, we conclude that for any $l \in \mNplus$ and $\delta \in \mRdplus$, there exists an associated $k_{l, \delta}$, s.t. for any $k \geq k_{l, \delta}$, we have
    \begin{equation}
    \label{eqn: pr_mut_inf_den_rate}
        \Pr\left\{\ZFXrXOpklOpPhi \geq \frac{1}{l} I\Big(\XvecEpOpPhiL; \rXvecEpOpPhiLOp\Big) + \delta \right\} < 3 \delta.
    \end{equation}

    Recalling the definition of the limit superior in probability in Def.~\ref{def: limsup_in_prob}, we have
    \begin{equation}
        \limsupp_{l \to \infty} \ZFXrXOpklOpPhi = \inf\left\{\alpha \in \mR \bigg| \lim_{l \to \infty}\Pr\Bigg\{\ZFXrXOpklOpPhi > \alpha\Bigg\} = 0 \right\}.
    \end{equation}
    Therefore, step~\eqref{step: limsupp_leq_rdf_ep_n} in Eqn.~\eqref{eqn: 2nd_gen_sche_achi_fin_exp} is proved if  for any $\delta \in \mRdplus$, the probability
    \begin{equation}
    \label{eqn: aim_achi_proof}
        \Pr\Bigg\{\ZFXrXOpklOpPhi > \limsup_{n \to \infty} R_{\epsilon_{n}}(D) + 5\delta \Bigg\},
    \end{equation}
    can be made arbitrarily small by the proper selection of the optimal initial sampling phase $\phisopt$ and the joint \gls{cdf} $F_{\XvecEpOpPhik, \rOpXvecEpOpPhik}$,  taking $l \in \mNplus$ and $k \in \mNplus$ sufficiently large. To that aim, define a constant $\tldgamma$ as
    \begin{equation}
    \label{eqn: def_tilde_gamma}
        \tldgamma \triangleq \frac{1}{2} \cdot \bigg(\log(\gamma) + \frac{\log(e)}{\gamma_{c}}\bigg),
    \end{equation}
    and select $l$ large enough s.t.
        $\frac{\tau_{c} \cdot \tldgamma}{l + \tau_{c}} < \delta$.
    Next, further increase  $l$ and select $n \in \mNplus$ sufficiently large s.t.
    \begin{equation}
    \label{eqn: delta_min_mut_inf_converg}
        \bigg|\frac{1}{l + \tau_{c}}I\Big(\XvecEpnPhiLop; \rXvecEpnPhiLop\Big) - \frac{1}{l} I\Big(\XvecEpOpPhiL; \rXvecEpOpPhiLOp\Big)\bigg| < \delta,
    \end{equation}
    where the pairs $\bigg(\phi_{s, \epsilon_{n}, l}^{\opt}, p\Big(\rXvecEpnPhiLop | \XvecEpnPhiLop\Big)\bigg)$ and $\bigg(\phi_{s, \epsilon, l}^{\opt}, p\Big(\rXvecEpPhiLop | \XvecEpPhiLop\Big)\bigg)$ are obtained from the minimization in Eqns.~\eqref{eqn: op_pair_mut_inf_ep_n} and \eqref{eqn: op_pair_mut_inf_ep}, respectively. Note that due to the convergence in Lemma~\ref{lem: op_mut_inf_converg} and $\lim_{l \to \infty} \frac{l}{l + \tau_{c}} = 1$, it is possible to find a pair $l$, $n$ s.t. Eqn.~\eqref{eqn: delta_min_mut_inf_converg} is satisfied. Next, we further increase $n$ to guarantee
    \begin{equation}
        \label{eqn: n_long_enough}
        R_{\epsilon_{n}}(D) < \limsup_{n_{0} \to \infty} R_{\epsilon_{n_{0}}}(D) + \delta,
    \end{equation}
    which is possible by the definition of the limit superior. Fixing $n$, We then pick $k \in \mNplus$ large enough s.t.
    \begin{equation}
    \label{eqn: asymp_achi_rdp}
        R_{\epsilon_{n}}(D) \geq \frac{1}{k \cdot (l + \tau_{c})} I\Big(\XvecEpnPhiGIkOp; \rXvecEpnPhiGIkOp\Big) - \delta,
    \end{equation}
    $\forall \PhiEpn \in [0, T_{c})$, where the reconstruction process $\mr{X}_{\epsilon_{n}}^{\PhiEpn, \opt}[i]$, over each $l$ symbols of an $(l+\tau_c)$-block follows the optimal distribution given the \gls{dt} \gls{wscs} process $X_{\epsilon_{n}}^{\PhiEpn}[i]$, in the sense that it achieves the \gls{rdf} in \cite[Thm.~1]{kipnis2018}. We note that such selection of $k$ is possible by 
    the definition of asymptotically achievable rate-distortion pairs, see \cite[Def.~8.10]{yeung2008}.

    Denote the $\tau_{c}$ symbols appended after the $m$-th $l$-block for the synchronous sampling scenario by $\tauSymEpnPhim$ and denote its reconstruction by $\rtauSymEpnPhim \equiv \zerovec^{\tau_{c}}$, $\PhiEpn \in [0, T_{c})$, $0 \leq m \leq k - 1$. Thus, the $k \cdot \tau_{c}$ symbols representing the guard intervals are reconstructed as zero samples at the decoder. We then define four random vectors $\XvecOne$, $\XvecTwo$, $\rXvecOne$ and $\rXvecTwo$ for synchronous sampling cases as follows:
        \begin{align*}
            \XvecOne & \equiv \Big\{\XvecEpnPhim\Big\}_{m = 0}^{k - 1}, \qquad
            \XvecTwo \equiv \Big\{\tauSymEpnPhim\Big\}_{m = 0}^{k - 1}, \\
            \rXvecOne & \equiv \Big\{\rXvecEpnPhimOp\Big\}_{m = 0}^{k - 1}, \qquad 
            \rXvecTwo \equiv \Big\{\rtauSymEpnPhim\Big\}_{m = 0}^{k - 1} \equiv \zerovec^{k \cdot \tau_{c}}.
        \end{align*}
    It is noted that $\big[(\XvecTwo)^{T} (\XvecOne)^{T}\big]^{T}$ is the permutation of the vector $\XvecEpnPhiGIk$. Letting  $\permmat$ denote the permutation matrix, we write $\big[(\XvecTwo)^{T} (\XvecOne)^{T}\big]^{T} = \permmat \cdot \XvecEpnPhiGIk$. Accordingly, we define $\big[(\SvecTwo)^{T} (\SvecOne)^{T}\big]^{T} \triangleq \big[(\XvecTwo)^{T} (\XvecOne)^{T}\big]^{T} - \big[(\rXvecTwo)^{T} (\rXvecOne)^{T}\big]^{T} = \permmat \cdot \Big(\XvecEpnPhiGIk - \rXvecEpnPhiGIkOp\Big) \triangleq \permmat \cdot \OpSvecEpPhiGIk$. 
    
    With these definitions, we obtain equalities presented as follows:
        \begin{align}
            & I\Big(\XvecEpnPhiGIkOp; \rXvecEpnPhiGIkOp\Big) \nonumber \\
            & \equiv I(\XvecOne, \XvecTwo; \rXvecOne, \rXvecTwo) \nonumber \\
            & = I(\rXvecOne; \XvecOne, \XvecTwo) + I(\rXvecTwo; \XvecOne, \XvecTwo | \rXvecOne) \nonumber \\
            & = I(\XvecOne; \rXvecOne) + I(\rXvecOne; \XvecTwo | \XvecOne) + I(\rXvecTwo; \XvecOne, \XvecTwo | \rXvecOne) \nonumber \\
            & = I(\XvecOne; \rXvecOne) + h(\XvecTwo | \XvecOne) - h(\XvecTwo | \XvecOne, \rXvecOne) + h(\XvecOne, \XvecTwo | \rXvecOne) \nonumber \\
            & \quad - h(\XvecOne, \XvecTwo | \rXvecOne, \rXvecTwo) \nonumber \\
            & = I(\XvecOne; \rXvecOne) + h(\XvecTwo | \XvecOne) - h(\XvecTwo | \XvecOne, \rXvecOne) + h(\XvecOne | \rXvecOne ) \nonumber \\
            & \quad + h(\XvecTwo |\XvecOne, \rXvecOne ) - h(\XvecOne, \XvecTwo | \rXvecOne, \rXvecTwo) \nonumber \\
            & = I(\XvecOne; \rXvecOne) + h(\XvecTwo | \XvecOne) + h(\XvecOne | \rXvecOne) - h(\XvecOne, \XvecTwo | \rXvecOne, \rXvecTwo) \nonumber \\
            & \labelrel={step: dist_diff} I(\XvecOne; \rXvecOne) + h(\XvecTwo | \XvecOne) + h(\SvecOne | \rXvecOne) - h(\SvecOne, \SvecTwo | \rXvecOne, \rXvecTwo) \nonumber \\
            & \labelrel={step: dist_ind} I(\XvecOne; \rXvecOne) + h(\XvecTwo | \XvecOne) + h(\SvecOne) - h(\SvecOne, \SvecTwo) \nonumber \\
            & \label{eqn: mut_inf_with_guard_intvls} = I(\XvecOne; \rXvecOne) - \big(h(\SvecTwo | \SvecOne) - h(\XvecTwo | \XvecOne)\big),
        \end{align}
    where \eqref{step: dist_diff} follows as $\big[(\SvecTwo)^{T} (\SvecOne)^{T}\big]^{T} \triangleq \big[(\XvecTwo)^{T} (\XvecOne)^{T}\big]^{T} - \big[(\rXvecTwo)^{T} (\rXvecOne)^{T}\big]^{T}$; and \eqref{step: dist_ind} follows from the statistical independence between $\big[(\SvecTwo)^{T} (\SvecOne)^{T}\big]^{T}$ and $\big[(\rXvecTwo)^{T} (\rXvecOne)^{T}\big]^{T}$ (see \cite[Sec.~10.3.2]{cover2006}). 
    
    Define $\macrmat{\XvecTwo \cdot \XvecOne} \triangleq \mE\{\XvecTwo \cdot (\XvecOne)^{T}\}$, $\macrmat{\XvecOne \cdot \XvecTwo} \triangleq \mE\{\XvecOne \cdot (\XvecTwo)^{T}\}$ and $\macrmat{\XvecTwo, \XvecOne} \triangleq \mE\big\{[(\XvecTwo)^{T} (\XvecOne)^{T}]^{T} \cdot [(\XvecTwo)^{T} (\XvecOne)^{T}] \big\}$. Denote the maximal diagonal element of the real square matrix $\mmat{A}$ by $\maxDiag\{\mmat{A}\}$. 
    We can upper bound $h(\SvecTwo | \SvecOne) - h(\XvecTwo | \XvecOne)$ as follows:
    \begin{align*}
        & h(\SvecTwo | \SvecOne) - h(\XvecTwo | \XvecOne) \\
        & \labelrel\leq{step: cond_reduc_diff_entrop} h(\XvecTwo) - h(\XvecTwo | \XvecOne) \\
        & \labelrel={step: gaussian_dist} \frac{1}{2} \cdot \bigg( \log\det\big(2\pi e \cdot \macrmat{\XvecTwo}\big) - \log\det\Big(2\pi e \cdot \big(\macrmat{\XvecTwo} - \macrmat{\XvecTwo \cdot \XvecOne} \cdot (\macrmat{\XvecOne})^{-1} \cdot \macrmat{\XvecOne \cdot \XvecTwo}\big) \Big) \bigg) \\
        & \labelrel\le{step: hadamard_inequ} \frac{1}{2} \cdot \Big( k \cdot \tau_{c} \cdot \log(\gamma) - \log\det\big(\macrmat{\XvecTwo} - \macrmat{\XvecTwo \cdot \XvecOne} \cdot (\macrmat{\XvecOne})^{-1} \cdot \macrmat{\XvecOne \cdot \XvecTwo}\big) \Big) \\
        & \labelrel\leq{step: log_det_inequ} \frac{1}{2} \cdot \Bigg( k \cdot \tau_{c} \cdot \log(\gamma) + \log(e) \cdot \bigg(\tr\Big\{\big(\macrmat{\XvecTwo} - \macrmat{\XvecTwo \cdot \XvecOne} \cdot (\macrmat{\XvecOne})^{-1} \cdot \macrmat{\XvecOne \cdot \XvecTwo}\big)^{-1}\Big\} - k \cdot \tau_{c}\bigg) \Bigg) \\
        & \leq \frac{k \cdot \tau_{c}}{2} \cdot \bigg(\log(\gamma) + \log(e) \cdot \maxDiag\Big\{\big(\macrmat{\XvecTwo} - \macrmat{\XvecTwo \cdot \XvecOne} \cdot (\macrmat{\XvecOne})^{-1} \cdot \macrmat{\XvecOne \cdot \XvecTwo}\big)^{-1}\Big\}\bigg) \\
        & \labelrel\leq{step: maxdiag_thm} \frac{k \cdot \tau_{c}}{2} \cdot \Big(\log(\gamma) + \log(e) \cdot \maxDiag\big\{(\macrmat{\XvecTwo, \XvecOne})^{-1}\big\}\Big) \\
        & \labelrel\leq{step: permmat_exp} \frac{k \cdot \tau_{c}}{2} \cdot \left(\log(\gamma) + \log(e) \cdot \Bigg\Vert \bigg(\permmat \cdot \macrmat{\XvecEpnPhiGIk} \cdot \permmat^{T}\bigg)^{-1} \Bigg\Vert_{1} \right) \\
        & \labelrel\leq{step: permmat_ortho_norm_prop} \frac{k \cdot \tau_{c}}{2} \cdot \left(\log(\gamma) + \log(e) \cdot \Bigg\Vert \bigg(\macrmat{\XvecEpnPhiGIk}\bigg)^{-1} \Bigg\Vert_{1} \right) \\
        & \labelrel\leq{step: by_sdd} \frac{k \cdot \tau_{c}}{2} \cdot \bigg(\log(\gamma) + \frac{\log(e)}{\gamma_{c}}\bigg),
    \end{align*}
    where \eqref{step: cond_reduc_diff_entrop} follows as $\SvecTwo \triangleq \XvecTwo - \rXvecTwo = \XvecTwo - \zerovec^{k \cdot \tau_{c}} = \XvecTwo$; \eqref{step: gaussian_dist} follows from the Gaussianity of the vector $\XvecTwo$ and as the conditional distribution of $\XvecTwo$ given $\XvecOne$ is Gaussian with the autocovariance matrix $\big(\macrmat{\XvecTwo} - \macrmat{\XvecTwo \cdot \XvecOne} \cdot (\macrmat{\XvecOne})^{-1} \cdot \macrmat{\XvecOne \cdot \XvecTwo}\big)$ (see \cite[Sec.~21.6]{fristedt1997}); \eqref{step: hadamard_inequ} follows from the Gaussianity of the vector $\XvecTwo$ and Hadamard's inequality \cite[Eqn.~(8.64)]{cover2006}, and as the \gls{af} of the \gls{ct} source process $X_{c}(t)$ is bounded by $\gamma$ (see Section~\ref{sec: prob_form_mod}); \eqref{step: log_det_inequ} follows from the symmetric positive definiteness of $\big(\macrmat{\XvecTwo} - \macrmat{\XvecTwo \cdot \XvecOne} \cdot (\macrmat{\XvecOne})^{-1} \cdot \macrmat{\XvecOne \cdot \XvecTwo}\big)^{-1}$ and \cite[Lemma~11.6]{golub2010}, and from the fact $\log(x) = \ln(x) \cdot \log(e) \leq (x - 1) \cdot \log(e)$ for $x \in \mRdplus$\footnote{As the matrix $\big(\macrmat{\XvecTwo} - \macrmat{\XvecTwo \cdot \XvecOne} \cdot (\macrmat{\XvecOne})^{-1} \cdot \macrmat{\XvecOne \cdot \XvecTwo}\big)^{-1}$ is symmetric positive definite, it is diagonalizable and all its eigenvalues are positive. Thus, this fact is applicable.}; \eqref{step: maxdiag_thm} follows from \cite[Eqn.~(0.7.3.1)]{horn2012}, which implies that $\big(\macrmat{\XvecTwo} - \macrmat{\XvecTwo \cdot \XvecOne} \cdot (\macrmat{\XvecOne})^{-1} \cdot \macrmat{\XvecOne \cdot \XvecTwo}\big)^{-1}$ is the upper-left block of $(\macrmat{\XvecTwo, \XvecOne})^{-1}$; \eqref{step: permmat_exp} follows from the definition of the $1$-norm of a square matrix, which is not smaller than its maximal diagonal element, and as $\macrmat{\XvecTwo, \XvecOne} \triangleq \mE\big\{[(\XvecTwo)^{T} (\XvecOne)^{T}]^{T} \cdot [(\XvecTwo)^{T} (\XvecOne)^{T}] \big\} = \permmat \cdot \mE\bigg\{\XvecEpnPhiGIk \cdot \Big(\XvecEpnPhiGIk\Big)^{T}\bigg\} \cdot \permmat^{T} = \permmat \cdot \macrmat{\XvecEpnPhiGIk} \cdot \permmat^{T}$; \eqref{step: permmat_ortho_norm_prop} follows from the orthogonality of permutation matrices \cite[Sec.~0.9.5]{horn2012} and the submultiplicativity of the $1$-norm of matrices (see \cite[Pg.~341 and Example~5.6.4]{horn2012}), which result in $\bigg(\permmat \cdot \macrmat{\XvecEpnPhiGIk} \cdot \permmat^{T}\bigg)^{-1} = (\permmat^{T})^{-1} \cdot \bigg(\macrmat{\XvecEpnPhiGIk}\bigg)^{-1} \cdot \permmat^{-1} = \permmat \cdot \bigg(\macrmat{\XvecEpnPhiGIk}\bigg)^{-1} \cdot \permmat^{T}$ and $\bigg\Vert \permmat \cdot \bigg(\macrmat{\XvecEpnPhiGIk}\bigg)^{-1} \cdot \permmat^{T} \bigg\Vert_{1} \leq \Vert \permmat \Vert_{1} \cdot \bigg\Vert \bigg(\macrmat{\XvecEpnPhiGIk}\bigg)^{-1} \bigg\Vert_{1} \cdot \Vert \permmat^{T} \Vert_{1} = \bigg\Vert \bigg(\macrmat{\XvecEpnPhiGIk}\bigg)^{-1} \bigg\Vert_{1}$\footnote{As $\permmat$ has only one element of $1$ in each row and in each column with all other elements 0, $\Vert \permmat \Vert_{1} = \Vert \permmat^{T} \Vert_{1} = 1$.}, respectively; and \eqref{step: by_sdd} follows from the similar derivation from Eqns.~\eqref{eqn: inv_1_norm_inv_autocorr_mat} to \eqref{eqn: low_boud_min_eig}. Recalling the definition of $\tldgamma$ in Eqn.~\eqref{eqn: def_tilde_gamma}, we obtain the upper bound of $h(\SvecTwo | \SvecOne) - h(\XvecTwo | \XvecOne)$ as
    \begin{equation}
        \label{eqn: diff_entrop_2_1_up_bound}
        h(\SvecTwo | \SvecOne) - h(\XvecTwo | \XvecOne) \leq k \cdot \tau_{c} \cdot \tldgamma.
    \end{equation}
    Plugging the upper bound in Eqn.~\eqref{eqn: diff_entrop_2_1_up_bound} back into Eqn.~\eqref{eqn: mut_inf_with_guard_intvls}, we obtain
        \begin{align}
            & I\Big(\XvecEpnPhiGIkOp; \rXvecEpnPhiGIkOp\Big) \\
            & \geq I(\XvecOne; \rXvecOne) - k \cdot \tau_{c} \cdot \tldgamma \nonumber \\
            & \geq I\bigg(\Big\{\XvecEpnPhim\Big\}_{m = 0}^{k - 1}; \Big\{\rXvecEpnPhimOp\Big\}_{m = 0}^{k - 1}\bigg) \\
            & \label{eqn: low_boud_mut_inf_with_gi} \quad\quad\quad\quad\quad\quad\quad\quad\quad\quad\quad\quad - k \cdot \tau_{c} \cdot \tldgamma.
        \end{align}
    Then, given the average distortion $D$, the bound on the \gls{rdf} $R_{\epsilon_{n}}(D)$ in  Eqn.~\eqref{eqn: asymp_achi_rdp} can be relaxed as follows:
        \begin{align}
            R_{\epsilon_{n}}(D) & \geq \frac{1}{k \cdot (l + \tau_{c})} \cdot I\Big(\XvecEpnPhiGIkOp; \rXvecEpnPhiGIkOp\Big) - \delta \nonumber \\
            & \nonumber \labelrel\geq{step: intro_low_bound_mut_inf} \frac{1}{k \cdot (l + \tau_{c})} \cdot I\bigg(\Big\{\XvecEpnPhim\Big\}_{m = 0}^{k - 1}; \Big\{\rXvecEpnPhimOp\Big\}_{m = 0}^{k - 1}\bigg) - \frac{k \cdot \tau_{c} \cdot \tldgamma}{k \cdot (l + \tau_{c})} - \delta \\
            & \nonumber \labelrel={step: val_delta} \frac{1}{k \cdot (l + \tau_{c})} \cdot \sum_{m = 0}^{k - 1} I\Big(\XvecEpnPhim; \rXvecEpnPhimOp\Big) - 2\delta \\
            & \label{eqn: rdf_epn_low_bound} \labelrel\geq{step: opt_init_samp_phase} \frac{1}{l + \tau_{c}} \cdot I\Big(\XvecEpnPhiLop; \rXvecEpnPhiLop\Big) - 2\delta,
        \end{align}
    where \eqref{step: intro_low_bound_mut_inf} follows by plugging Eqn.~\eqref{eqn: low_boud_mut_inf_with_gi}; \eqref{step: val_delta} follows from the statistical independence between different $l$-blocks and as $l$ is selected  s.t. $\frac{\tau_{c} \cdot \tldgamma}{l + \tau_{c}} < \delta$; and \eqref{step: opt_init_samp_phase} follows from the minimization of the mutual information w.r.t. the initial sampling phase of each $l$-block over $[0, T_{c})$ (recall the minimization in Eqn.~\eqref{eqn: op_pair_mut_inf_ep_n}) and as all $l$-blocks are \gls{iid}.
    
    Lastly, recalling the probability in Eqn.~\eqref{eqn: aim_achi_proof}, for any $\delta \in \mRdplus$, we can properly select $l$, $n$ and $k$ s.t.
        \begin{align*}
            & \Pr\Bigg\{\ZFXrXOpklOpPhi > \limsup_{n \to \infty} R_{\epsilon_{n}}(D) + 5\delta\Bigg\} \\
            & \labelrel\leq{step: n_suff_large} \Pr\Bigg\{\ZFXrXOpklOpPhi > R_{\epsilon_{n}}(D) + 4\delta\Bigg\} \\
            & \labelrel\leq{step: R_epn_low_boud} \Pr\Bigg\{\ZFXrXOpklOpPhi > \frac{1}{l + \tau_{c}} \cdot I\Big(\XvecEpnPhiLop; \rXvecEpnPhiLop\Big) + 2\delta\Bigg\} \\
            & \labelrel\leq{step: mut_inf_ep_n_converg} \Pr\Bigg\{\ZFXrXOpklOpPhi > \frac{1}{l} \cdot I\Big(\XvecEpOpPhiL; \rXvecEpOpPhiLOp\Big) + \delta\Bigg\} \\
            & \labelrel\leq{step: pr_mut_inf_den_rate_exp} 3\delta,
        \end{align*}
    where \eqref{step: n_suff_large} follows as $n$ is selected to be sufficiently large s.t. the condition in Eqn.~\eqref{eqn: n_long_enough} is satisfied; \eqref{step: R_epn_low_boud} follows from Eqn.~\eqref{eqn: rdf_epn_low_bound}; \eqref{step: mut_inf_ep_n_converg} follows from Eqn.~\eqref{eqn: delta_min_mut_inf_converg}; and lastly \eqref{step: pr_mut_inf_den_rate_exp} follows from Eqn.~\eqref{eqn: pr_mut_inf_den_rate}. Taking $l \to \infty$ and $\delta \to 0$, conclude
    \begin{equation*}
        \lim_{l \to \infty} \Pr\Bigg\{\ZFXrXOpklOpPhi > \limsup_{n \to \infty} R_{\epsilon_{n}}(D) \Bigg\} = 0.
    \end{equation*}
    Recalling the scaling factor $\Bigg(1 - \frac{\tau_{c} + \frac{\Delta_{g}}{T_{s}(\epsilon)}}{l + \tau_{c} + \frac{\Delta_{g}}{T_{s}(\epsilon)}}\Bigg)$ in Eqn.~\eqref{eqn: overall_code_rate}, 
    an actual code rate of $\limsup_{n \to \infty} R_{\epsilon_{n}}(D) \cdot \Bigg(1 - \frac{\tau_{c} + \frac{\Delta_{g}}{T_{s}(\epsilon)}}{l + \tau_{c} + \frac{\Delta_{g}}{T_{s}(\epsilon)}}\Bigg)$ is achievable.  
    Taking $l$  sufficiently large, the asymptotically achievable code rate of $\limsup_{n \to \infty} R_{\epsilon_{n}}(D)$ is obtained. It is finally noted that in this proof we only consider the input codewords whose blocklengths are integer multiples of $l$. As $l$ is fixed, by taking $k$ sufficiently large and padding at most $l - 1$ sampled source symbols, a codebook with an arbitrary blocklength can be obtained with an asymptotically negligible decrease of code rate. This completes the proof of step~\eqref{step: limsupp_leq_rdf_ep_n} in Eqn.~\eqref{eqn: 2nd_gen_sche_achi_fin_exp}.

    Combining the lower bound in Eqn.~\eqref{eqn: converse_2nd_scheme} and the upper bound in Eqn.~\eqref{eqn: 2nd_gen_sche_achi_fin_exp}, it is concluded that $R_{\epsilon}(D) = \limsup_{n \to \infty} R_{\epsilon_{n}}(D)$. This completes the proof of Lemma~\ref{lem: final_lem_2nd_scheme} and therefore completes the proof of Thm.~\ref{thm: main_thm}.
\end{proof}

\end{appendices}

\bibliographystyle{bst._files/self_adjusted_IEEEtran.bst}   
\bibliography{IEEEabrv, references.bib}

\end{document}